\documentclass[12pt,leqno]{amsart}
\usepackage{amssymb,amsmath,amsthm}
\usepackage{graphicx}
\def\e{{\rm e}}
\def\eps{\varepsilon}

\def\d{{\rm d}}

\def\dim{{\rm dim}}

\def\ddt{\frac{\d}{\d t}}

\def\AA {{\mathfrak A}}

    \def\LL {{\mathfrak L}}

\def\R {\mathbb{R}}
\def\N {\mathbb{N}}
\def\H {{\mathcal H}}
\def\V {{\mathcal V}}

\def\B {{\mathcal B}}
\def\C {{\mathcal C}}

\def\E {{\mathcal E}}
\def\A {{\mathcal A}}
\def\AA {{\mathbb A}}

\def\I {{\mathcal I}}

\def\LL {{\mathcal L}}
\def\Q {{\mathcal Q}}
\def\U {{\mathcal U}}

\def \l {\langle}
\def \r {\rangle}
\def \pt {\partial_t}
\def \ptt {\partial_{tt}}

\def \and{\quad\text{and}\quad}

\def \au {\rm}
\def \ti {\it}
\def \jou {\rm}

\def \no#1#2#3 {{\bf #1} (#3), #2.}
\def \eds#1#2#3 {#1, #2, #3.}

\newtheorem{proposition}{Proposition}[section]

\newtheorem{theorem}{Theorem}[section]
\newtheorem{corollary}{Corollary}[section]
\newtheorem{lemma}{Lemma}[section]
\theoremstyle{definition}

\newtheorem{remark}{Remark}[section]
\newtheorem*{remark*}{Remark}
\newtheorem*{warn*}{A word of warning}

\numberwithin{equation}{section}

\title[Time-Dependent Attractor for the Oscillon Equation ]
{Time-Dependent Attractor for The Oscillon Equation}

\author[F.\ Di Plinio]
{Francesco Di Plinio}
\address{
Institute for Scientific Computing and Applied Mathematics
\newline\indent
Indiana University
\newline\indent
Bloomington, IN 47405 - USA}
\email{fradipli@indiana.edu {\rm (F.\ Di Plinio)} }
\email{temam@indiana.edu {\rm (R.\ Temam)} }

\author[G.\ S.\ Duane]
{Gregory S.\ Duane}

\address{
Rosenstiel School of Marine and Atmospheric Sciences
\newline\indent
University of Miami
\newline\indent
Miami, FL 33149
\newline\indent \vskip-3mm
Dept. of Atmospheric and Oceanic Sciences
\newline\indent
University of Colorado
\newline\indent
Boulder, CO 80309
}
\email{gregory.duane@colorado.edu {\rm (G.\ S.\ Duane)} }

\author[R.\ Temam]
{Roger Temam}

\date{\today}
\begin{document}

\begin{abstract}We investigate the asymptotic behavior of the nonautonomous evolution problem generated by
the Oscillon equation
$$
\ptt u(x,t) +H \pt u(x,t) -\e^{-2Ht}\partial_{xx} u(x,t) +
V'(u(x,t)) =0, \qquad x\in (0,1), t \in \R,
$$
with periodic boundary conditions, where $H>0$ is the Hubble
constant and $V$ is a nonlinear potential of arbitrary polynomial
growth.  After constructing a suitable dynamical framework to deal
with the explicit time dependence of the energy of the solution,
we establish the existence  of a regular global attractor
$\A=\A(t)$. The kernel sections $\A(t)$ have finite fractal dimension.
\end{abstract}

\maketitle

\section{Introduction}

The present paper is devoted to the study of the nonautonomous evolution problem
generated by the equation
\begin{equation} \label{EQ-INTRO}
\ptt u(x,t) +H \pt u(x,t) -\e^{-2Ht}\partial_{xx} u(x,t) + V'(u(x,t)) =0,
\qquad x\in (0,1), t \in \R,
\end{equation}
where $V$ is a nonlinear potential of polynomial growth satisfying
natural dissipativity conditions. Equation \eqref{EQ-INTRO} is the
Klein-Gordon equation, with the given nonlinear potential, for a
scalar field on a manifold with a Robertson-Walker metric
corresponding to an expanding universe with Hubble constant $H>0$,
and is referred to here as the \emph{oscillon} equation.

 The physical motivation for the mathematical development in this paper stems from
a role of Equation (\ref{EQ-INTRO}) in recent cosmological
theories. It has been suggested that long-lived, localized,
oscillating solutions to the equations of the standard model of
particle physics may be useful in breaking thermodynamic
equilibrium.
\begin{figure}
\begin{center}
\includegraphics[width=\textwidth]{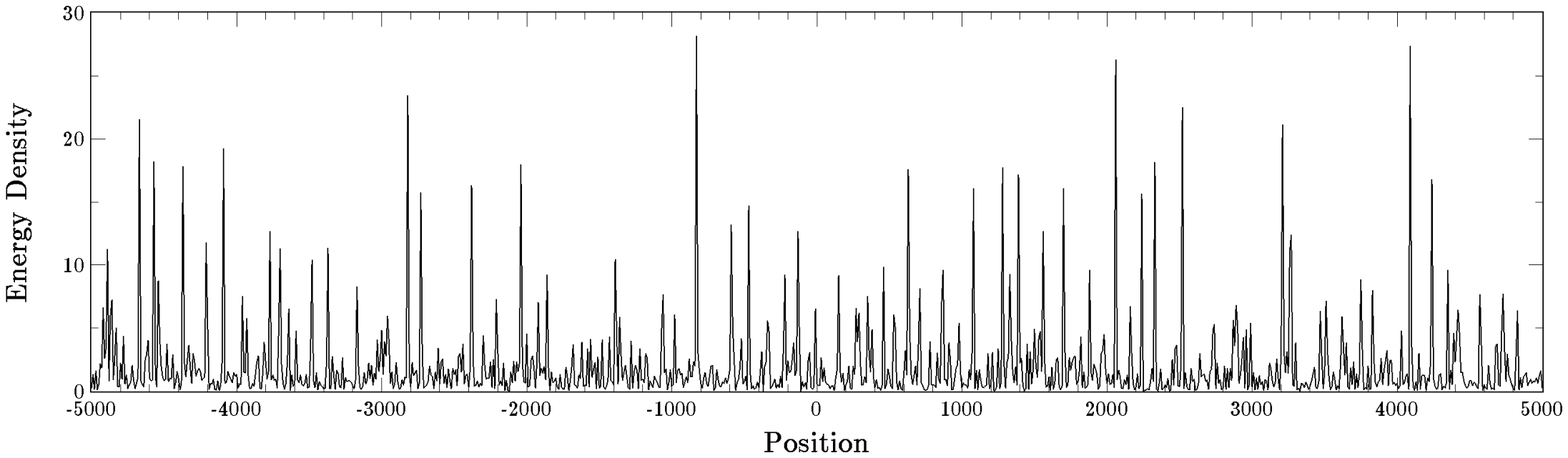}
\includegraphics[width=\textwidth]{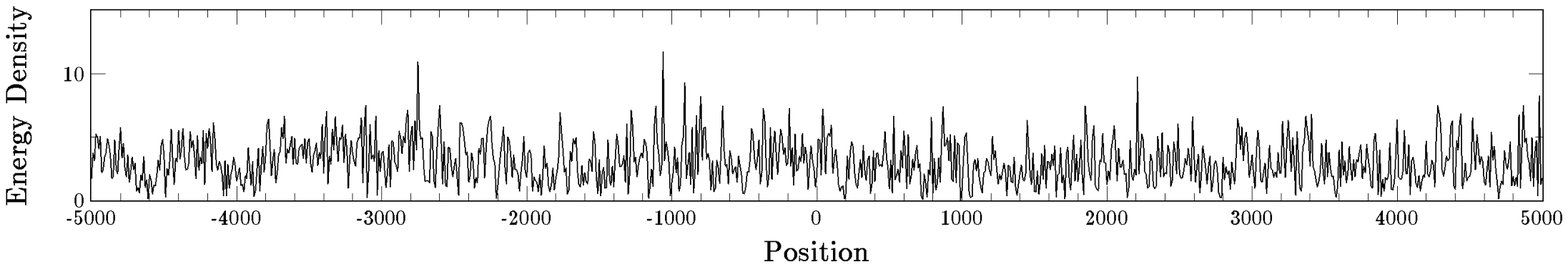}
\caption{Energy density vs. position $x$ for a numerical
simulation of (\ref{EQ-INTRO}) with periodic boundary conditions
and $V(y)=\frac12 y^2 - \frac14y^4+ \frac16y^6$, suggesting
localized oscillations   (top), and a simulation of the same
equation, but with a different potential $V(y)= \frac12 y^2 +
\frac14y^4+ \frac16y^6$, for which oscillons do not occur, shown
for comparison (bottom). Oscillons occur for the first potential
and not the second for the following physical reason: the $-y^4$
term causes large amplitude perturbations to see a flatter
potential and thus oscillate at lower frequency. The low-frequency
oscillations (oscillons) decouple from the higher-frequency,
lower-amplitude travelling waves that would otherwise cause them
to dissipate. (The $y^6$ term is included to maintain a positive
potential and avoid unbounded oscillations). (from a 9/26/05
presentation to DOE by A.\ Guth, with E.\ Farhi, R.\ Rosales, N.\
Graham, A.\ Scardicchio, and R. Stowell)} \label{figoscillon}
\end{center}
\end{figure}
Long-lived ``oscillons" have been found to occur in a numerical
simulation of the simplified model considered here, in one space
dimension,   for a nonlinear potential $V$ of appropriate form
\cite{FGKet.al.}, as shown in Fig.\ \ref{figoscillon}.  Oscillons
have also been studied in the context of other simple models of
early universe phase transitions \cite{OS1,OS4,OS2,OS3}. These
coherent structures seem to capture the essential features of the
oscillon phenomenon in simulations of the full standard model in
three-dimensional space.   Localized oscillating solutions had
been previously found in models of vibrating granular media
\cite{Um}. In the cosmological context, the external forcing is
replaced by thermalization of initial conditions, and the
dissipation is scale-independent. The ``friction" term, containing
the first time-derivative of the field, results from the
expansion.

The structures seen in Fig.\ \ref{figoscillon}a are
localized, low-frequency oscillations.  The equation is written
in a coordinate system in which the expansion of the universe is not apparent.
Rather, the oscillons, which are of constant physical width, steadily decrease
in width in the chosen units.
The system (\ref{EQ-INTRO}) is formally Hamiltonian, with
time-dependent Hamiltonian density
\begin{equation}
\label{Ham}
{\H }(t) = \frac12 \e^{-Ht} [(\partial_x u)^2 + \pi^2] + \e^{Ht} V(u),
\end{equation}
where $\pi \equiv \e^{Ht} u$ is the canonical momentum that is conjugate to $u$.
That is, we can write the canonical equations:
$$
\dot u=\frac{\partial{\H}}{\partial\pi}, \qquad   \\
\dot\pi=-\frac{\partial{\H}}{\partial u}.
$$
 Liouville's theorem, which applies even when the Hamiltonian is
explicitly time-dependent, precludes the collapse of the whole phase
space to a lower-dimensional manifold, so that a standard
attractor cannot exist for the oscillon equation.  This paper demonstrates that, by suitably adapting the notions of dissipativity and of \textit{basin of attraction}, an alternative construct, the {\it pullback
attractor}, is able to capture  the intuitive notion that oscillons define a
restricted portion of the full phase space.

Well-posedness  of \eqref{EQ-INTRO}, supplemented with ordinary
(e.g.\ periodic or Dirichlet) boundary conditions,  is classical.
In fact,  to this regard   \eqref{EQ-INTRO} can be seen as  a
nonlinear damped wave equation in space dimension 1 with
time-dependent speed of propagation $\e^{-Ht}$. It is easy to see
that it generates a nonautonomous dynamical system
(\emph{process})  on the  energy phase space $(u,\pt u) \in
H^1(0,1) \times L^2(0,1)$ (see references). Therefore, the main
concern of the paper is the study of the dissipative properties
and asymptotic behavior of the process. In particular, we are
interested in   the construction of some sort of attractor, i.e.\
a ``thin'' (compact, possibly finite-dimensional) subset of the
phase space which is invariant under the process and  embodies the
long-term dynamics of the system.

While the theory of attractors for autonomous systems is
well-established (see \cite{BV,HAL,Har,LAD,TEM} for theoretical
background and classical applications), there is less common
ground in the nonautonomous case. The \emph{uniform attractor}
approach, initiated by Haraux \cite{Har}, and further developed by
Chepyzhov and Vishik \cite{CV1,CV2}, relies essentially on the
compactness of the nonautonomous terms in the equation, and
therefore is not applicable in our case.

On the contrary, the framework of pullback attractors, as
developed in \cite{Cra1,Cra2,Schm} (see also \cite{CLV,CLR,CKB}
for more recent investigations), does not pose essential
restrictions (for example, translation boundedness or translation
compactness) on the nonautonomous terms and allows for
time-dependent limit objects (absorbing families and attractors).
Roughly speaking, a pullback attractor $\A=\{\A(t)\}$ is a
(time-dependent) family of compact sets which is invariant under
the solution operator and attracts all solutions originating
sufficiently far in the past. The set $\A(t)$    describes the
\emph{regime} of the system at time $t$; the solutions starting
sufficiently early have forgotten the initial data and, thanks to
the invariance property, their future evolution is well described
by $\A$. In contrast to the uniform attractor framework, the sets
$\A(t)$ may very well  not be uniformly bounded as $t \to
+\infty$.

For our system,  the lack of dissipation for the  natural energy
$$
\E(t)=\e^{-2Ht} \int_0^1 |\partial_x u(x,t)|^2 \, \d x +   \int_0^1 |\pt u(x,t)|^2 \, \d x,
$$
due to the  singularities of the speed of propagation  for $t \to
\pm \infty$ and to the structure of the equation, prevents the
existence of a pullback absorbing set in the usual sense. For
example, in  the linear homogeneous  case $(V=0)$, one has the
conservation law
$$
\e^{2Ht}\E(t)= \e^{2Hs}\E(s),\;t,s \in \R,
$$
so that $\|\partial_ x u(t) \|_{L^2}$ approaches a constant
depending on the size of the data at time $s$, as $t \to +\infty$.

To circumvent these issues, we adopt a new point of view on
pullback dissipativity. \emph{The main idea is to restrict the
basin of attraction to those families of sets of the phase space
whose (time-dependent) energy $\E(t)$, dictated by the problem, is
bounded as time goes to  $- \infty.$} To this purpose, we describe
the solution operator as a family of maps acting on a
time-dependent family of spaces  $X_t$. In  our problem, the
spaces are all the same linear space   with the norms
$\|\cdot\|_{X_t}$, $\|\cdot\|_{X_s}$ equivalent for any fixed
$t,s$. However, as it will be clear below, this equivalence blows
up as we let $s,t \to \pm \infty$.
\subsection*{Plan of the paper}
In the subsequent Section 2, we describe an abstract framework for
dynamical processes on time-dependent spaces and  provide the main
existence result for  time-dependent global attractors.

In Section 3, we formulate the evolution problem generated by eq.
\eqref{EQ-INTRO} and provide the main dissipative estimate of the
paper. Section 4 is devoted to the construction of the
time-dependent attractor for the process corresponding to eq.
\eqref{EQ-INTRO}. In Section  5 we briefly describe the structure of the
attractor and some forward convergence results in the case of the  potentials coming from  the physics literature. In particular, possible nontriviality of the pullback attractor for very flat potentials, discussed in Subsection 5.1, points to a way in which the oscillon behavior of Figure \ref{figoscillon} might be explained in future extensions of the present work.  Finally, Section 6 is dedicated to establishing finite dimensional reduction on the time-dependent attractor.

\section{Attractors in Time-Dependent Spaces.} \label{sect:attractors}
As roughly described in the introduction, we need to modify the
classical framework of pullback attractors in order to handle
evolution problems, like \eqref{EQ-INTRO}, where the nonautonomous
terms appear at a \emph{functional} level, even acting on the
space derivatives of the highest order, and not merely as an
external time-dependent forcing. Therefore, the two parameter
solution operator will be described by  a family of maps acting on
a one parameter family of spaces $X_t$, which we continue to call
process.

\subsection*{Process.} For $t \in \R$, let $X_t$ be a family of Banach spaces.
A (continuous) \emph{process}  is a two-parameter family of
mappings $\{S(t,s): X_s \to X_t\}_{s \leq t}$ with properties
\begin{itemize}
  \item[(i)] $S(t,t) = \textrm{Id}_{X_t}$;
  \item[(ii)] $S(t,s) \in \C(X_s, X_t)$;
   \item[(iii)] $S(\tau,t)S(t,s) = S(\tau,s)$ for $s \leq t \leq \tau.$
\end{itemize}

In the concrete case examined in \S3-\S5, the spaces  $X_t$ are
all the same linear space with the norms $\|\cdot\|_{X_t}$,
$\|\cdot\|_{X_s}$ equivalent for any fixed $t,s$, whereas the
equivalence blows up as we let $s, t \to \pm \infty$. However,
this is not needed for most of the theory we develop hereafter,
and the present framework can handle evolution problems in which
the spaces $\{X_t\}$ are completely unrelated.

\subsection*{Pullback-bounded family.} A family of subsets $\B=\{\B(t) \subset X_t\}_{t \in \R}$
is \emph{pullback-bounded} if\footnote{Here, for $D$ subset of a
Banach space $X$, $\displaystyle \|D\|_{X} = \sup_{z \in D}
\|z\|_{X}$}
$$
R(t)=\sup_{s\in (-\infty, t]}\|\B(s) \|_{X_s} < \infty \qquad \forall t \in \R,
$$
i.e. the sets $\B(t)\subset X_t$ are bounded for all times and $\|\B(s) \|_{X_s}$ is bounded as $s \to -\infty$.
\vskip2mm

\subsection*{Pullback absorber.} A  pullback-bounded family  $\AA=\{\AA(t)\}$ is called \emph{pullback absorber} if for every pullback-bounded family $\B$ and for every $t \in \R$ there exists $t_0=t_0(t)\leq t$ such that
$$
S(t,s) \B(s) \subset \AA(t), \qquad \forall s \leq t_0.
$$
\vskip2mm
  Although the restriction of the basin of attraction to parametrized families of sets  has been employed before in the standard framework of pullback attractors (see for example \cite{CL,CLR,CLR0}), our definitions have a stronger physical connotation. It seems physically reasonable to assume that the collection of ``bounded sets'', which usually constitutes the basin of absorption, or attraction of a global attractor, contains only those families with bounded energy, as dictated by the problem.

\subsection*{Time-dependent global attractor.} A  family of compact subsets  $\A=\{\A(t) \subset X_t\}_{t \in \R}$ is called  \label{def:attractor}\emph{time-dependent global attractor} if it fulfills the following  properties:
\begin{itemize}
  \item[(i)](invariance) $S(t,s) \A(s) = \A(t)$, for every $s \leq t$ ;
  \item[(ii)] (pullback attraction) for every pullback-bounded family $\B$ and every $t \in \R$, \footnote{
For a Banach space $X$ and $A,B \subset X$, the Hausdorff semidistance is defined as
$$
\textrm{dist}_{X} (A,B) = \sup_{x \in A} \inf_ {y \in B} \|y-x\|_X.
$$
From the definition, $\textrm{dist}_{X} (A,B) =0$ if and only if $A$ is contained in the closure of $B$.
}   $$ \displaystyle \lim_{s \to - \infty } \textrm{dist}_{X_t} (S(t,s)\B(s),\A(t)) =0.$$
\end{itemize}
If property (ii) holds uniformly wrt $t \in \R$, $\A$ is a \emph{uniform} time-dependent global attractor.
\begin{remark} \label{uniqueness} In general, conditions (i)-(ii) are not sufficient for uniqueness of the time-dependent attractor \textbf{(see \cite{MZ} for a discussion on this issue and examples)}.  However, if we require in addition
\begin{itemize}
  \item[(iii)] $\A$ is a pullback-bounded family,
\end{itemize}
then there exists at most one family satisfying (i)-(iii), i.e. a
pullback-bounded time-dependent global attractor is unique in the
class of pullback-bounded families. Indeed, assume that $\A$ and
$\A'$ are two pullback-bounded time-dependent global attractors
for the process $S(\cdot,\cdot)$. Fix $t \in \R$ and observe that
(ii) implies
$$
\textrm{dist}_{X_t} (S(t,s)\A'(s),\A(t)) \to 0, \qquad s \to
-\infty.
$$
But $S(t,s) \A'(s) = \A'(t)$, so that $\A'(t) \subset \A(t)$, since $\A(t)$ is closed. Exchanging the roles of $\A(t)$ and $\A'(t)$ we get the reverse inclusion as well, so that $\A(t)=\A'(t),$ and finally $\A=\A'$.

\textbf{To obtain unconditional (i.e.\ without assuming (iii)) uniqueness of the time-dependent global attractor, one has to rule out the existence of families of compact, invariant, pullback-attracting sets which are not pullback-bounded: this can be done by establishing that  the process has the backward boundedness property.}
\end{remark}
\subsection*{Time-dependent $\omega$-limit.} Given a  family of sets $\B$, its time-dependent $\omega$-limit is the family
 $\omega_\B=\{\omega_\B(t) \subset X_t\}_{t \in \R}$, where $\omega_\B(t)$ is defined as
$$
\omega_\B(t)=\bigcap_{\tau \leq t} \overline{\bigcup_{s \leq \tau} S(t,s) \B(s)},
$$
and the above closure is  taken in $X_t$. An equivalent characterization is the following:
$$
\omega_\B(t) = \{z \in X_t : \exists\, s_n \to - \infty, z_n \in \B(s_n) \textrm{ with } \|S(t,s_n)z_n -z \|_{X_t} \to 0 \textrm{ as } n \to \infty\}.
$$

\begin{warn*} In the remainder of the paper, we will generally omit the phrase \emph{time-dependent} and hence say
\emph{global attractor} for  \emph{time-dependent global
attractor}.
\end{warn*}
\subsection*{Existence of the global attractor}
This subsection is devoted to a result of existence of a (time-dependent) global attractor.
We will make use of the Kuratowski measure of noncompactness: if $X$ is a Banach space and $D \subset X$,
$$
 \alpha(D) = \inf\{ d>0 : D \textrm{ has a finite cover of balls of } X \textrm{ of radius less than } d\}.
$$
Hereafter we recall some  properties of the Kuratowski measure we
are going to use, redirecting to \cite{HAL} for more details and
proofs:
\begin{itemize}
\item[(K.1)] $\alpha(D)=0$ if and only if $A$ is compact in $X$;
\item[(K.2)] $D_1 \subset D_2$ implies $\alpha(D_1) \leq \alpha(D_2)$;
\item[(K.3)] $\alpha(D)=\alpha(\overline{D})$;
\item[(K.4)] if $r_0 \in \R$ and $\{\U_r\}_{r \geq r_0}$ is a family of nonempty closed subsets of $X$ such that $$\U_{r_2} \subset \U_{r_1} \quad \forall r_2 > r_1 \geq r_0, \qquad \lim_{r \to +\infty} \alpha(\U_r)=0,$$ then $\U= \cap_{r \geq r_0} \U_r$ is nonempty and compact;
\item[(K.5)] let $\{\U_r\}_{r \geq r_0}$  be as in (K.4) and let  any two sequences  $r_n \to +\infty$,   $x_n \in \U_{r_n}$ be given; then $x_n$ possesses a subsequence converging to some $x \in \U$.
\end{itemize}

\begin{remark}
The shorthand $\alpha_t$ stands for the Kuratowski measure in the space $X_t$. We remark that, for fixed $s,t\in \R$, $\alpha_s$ and $\alpha_t$ are equivalent measures of noncompactness whenever  there is a Banach space isomorphism between $X_s$ and $ X_t$.
\end{remark}
\begin{theorem} \label{existence}
Assume that the process $S(\cdot,\cdot)$ possesses an absorber $\AA$ for which
\begin{equation} \label{alphaa}
\lim_{s \to  -\infty} \alpha_t(S(t,s)\AA(s)) =0, \qquad \forall t \in \R.
\end{equation}
Then, $\omega_\AA$
is a global attractor for $S(\cdot,\cdot).$
\end{theorem}

Before the proof of Theorem \ref{existence}, we make some remarks and derive two useful corollaries.

\begin{remark} We can compare Theorem \ref{existence} with the previous literature by examining the case $X_t = X$ (the Banach space is fixed), where $\A$ is the usual pullback (global) attractor. The existence result in e.g.\ \cite{Cra1} relies on the compactness
of the process, which is possible for parabolic
systems in a bounded domain. For hyperbolic systems, we mention  a criterion based on the \emph{pullback asymptotic compactness} of the solution operator  (\cite{CLR}, see also \cite{MRW} for a similar construction in the uniform case), successively employed for example in \cite{WAN}. In \cite[Theorem 3.10]{SCD}, the authors provide an abstract theorem which relies on the $\alpha$-contractivity of the process (a property which implies \eqref{alphaa}), though in their setting a nested pullback absorber (i.e. $\AA (s) \subset \AA(t)$ when $s \leq t$) is needed.   \end{remark}

\begin{remark} Taking advantage of the same techniques used in \cite{PZel2} for the construction of a global attractor for closed semigroups, the strong continuity assumption $S(t,s) \in \C(X_s,X_t)$ can be relaxed to the weaker
\begin{itemize}
\item[(iii)'] $S(t,s) : X_s \to X_t$ is a closed operator: given a convergent sequence $x_n \to x$ in $X_s$, if $S(t,s)x_n \to y$ in $X_t$, then $y=S(t,s)x.$
\end{itemize}
In particular, (iii)' holds whenever $S(t,s)$ is norm-to-weak continuous.
\end{remark}

The following corollary is a (more concrete) reformulation of
Theorem~\ref{existence}.
\begin{corollary}\label{cor:attractor}If the process $S(\cdot,\cdot)$ with absorber $\AA$ possesses a decomposition
$$
S(t,s) \AA(s) = P(t,s) + N(t,s)
$$
where
$$
\lim_{s \to -\infty } \|P(t,s) \|_{X_t} =0, \qquad \forall t \in \R
$$
and $N(t,s)$ is a compact subset of $X_t$ for all $t \in \R$ and $s \leq t$,
then  $\A(t)=\omega_\AA(t)$ is a global attractor for the process $S(\cdot, \cdot)$ .
\end{corollary}
Before proving Theorem~\ref{existence}, we dwell on  the  regularity and uniqueness of the  global attractor in a further corollary.
\begin{corollary}\label{cor:regularity}
Let $Y_t$ be a further family of Banach spaces satisfying, for
every $t \in \R$,\footnote{With $Y\Subset X$ we indicate compact
injection of the Banach space $Y$ into the Banach space $X$.}
\begin{itemize}
\item[$\cdot$] $Y_t \Subset X_t$;
\item[$\cdot$]  denoting with  $\I_t: Y_t \to X_t$ the canonical injection, the maps $\I_s$ are equibounded for $s \leq t$, i.e. $\sup_{s \leq t} \|\I_s\|_{\mathcal L(Y_s,X_s)} = C(t) < \infty$;
\item[$\cdot$]  closed balls of $Y_t$ are closed$\;$\footnote{For example, this holds when $Y_t$ is reflexive and compactly embedded into $X_t$.} in $X_t$ .
\end{itemize}
Under the same assumptions as in  Corollary~\ref{cor:attractor},
if in addition
$$
\sup_{s \in (-\infty,t]} \|N(t,s)\|_{Y_t} = h(t) < \infty \qquad \forall t \in \R,
$$
then the global attractor $\A$ is a   pullback-bounded family, and
henceforth unique in the sense of Remark \ref{uniqueness}.
Furthermore, it satisfies
$$\|\A(t)\|_{Y_t} \leq h(t)\qquad \forall t \in \R.$$
\end{corollary}
\begin{proof} Fix $t \in \R$  and $z \in \A(t)$. By definition, there exists  sequences $s_n \to - \infty$, $z_n \in \AA(s_n)$ such that $
\|S(t,s_n) z_n - z\|_{X_t} \to 0$ as $n \to \infty$. Using the decomposition of Corollary~\ref{cor:attractor},
$$
S(t,s_n)z_n = P_{z_n}(t,s_n) + N_{z_n}(t,s_n),
$$
with $ P_{z_n}(t,s_n) \in P(t,s_n),$ and $N_{z_n}(t,s_n)  \in N(t,s_n)$. In particular $\|N_{z_n}(t,s_n)\|_{Y_t} \leq h(t)$, i.e.\ the sequence $N_{z_n}(t,s_n)$ is contained in the  closed ball of  $Y_t$ with  radius $h(t)$, which we call $B_t$.  Now, using $\|P(t,s) \|_{X_t} \to 0$ as $s\to - \infty $
$$
\|N_{z_n}(t,s_n) - z\|_{X_t} \leq \|S(t,s_n) z_n - z\|_{X_t} + \|P_{z_n}(t,s_n)\|_{X_t} \to 0, \qquad n \to \infty.
$$
Therefore, $z$ is an accumulation point of $B_t$ (in the topology of $X_t$). By assumption, $B_t$ is closed in $X_t$, so that $z \in B_t$ as well.

This establishes that $\A(t) \subset B_t$, i.e.\  $\|\A(t)\|_{Y_t} \leq h(t)$ for every $t \in \R$, which in turn yields that $\A$ is a pullback-bounded family in $Y_t$. The second assumption yields the existence of $C=C(t)>0$ such that
$$
\|\A(s) \|_{X_s} \leq C(t) h(s), \qquad \forall s \leq t.
$$
Taking supremum over $s \leq t$, since  $h$ is increasing by definition, we obtain
$$
\sup_{s \in (-\infty,t]} \| \A(s)\|_{X_s} \leq C(t) h(t),
$$
i.e.\ $\A$ is a pullback-bounded family.  Uniqueness then follows from Remark \ref{uniqueness}.
\end{proof}

\begin{proof}[Proof of Theorem \ref{existence}.] We will prove that $\omega_\AA$, as defined above, is a family of compact sets satisfying (i)-(ii), and hence, is a global attractor for the process $S(t,s)$. We split the proof into two parts.
\vskip1.5mm \noindent $\diamond$ \hskip1mm Compactness and attraction property of $\omega_\AA(t)$.
\\ \noindent Let  $t \in \R$ be fixed, and $\eps >0$ be arbitrary. By assumption, there exists $t_0 \leq t$
  such that $\alpha_t(S(t,s)\AA(s)) < \eps$ whenever $s \leq t_0$. Then, since $\AA$ is an absorber, we can find $s_0 \leq t_0$ satisfying
$$
S(t_0,s)\AA(s) \subset \AA(t_0) \qquad \forall s \leq s_0.
$$
Therefore, for $\tau \leq s_0$,
$$
\U_\tau = \bigcup_{s \leq \tau} S(t,s) \AA(s) = \bigcup_{s \leq \tau} S(t,t_0)S(t_0,s) \AA(s)  \subset \bigcup_{s \leq \tau}S(t,t_0)\AA(t_0) = S(t,t_0)\AA(t_0),
$$
which yields $\alpha_t(\U_\tau) < \eps$ whenever $\tau \leq s_0$.
Therefore the sets $\overline{\U_\tau}$ are closed subsets of $X_t$, nested as $\tau \to -\infty$ and
$$
\lim_{\tau \to -\infty} \alpha_t(\overline{\U_\tau}) = 0.
$$
It follows by property (K.4) of the Kuratowski measure that
$
\omega_\AA(t)= \bigcap_{\tau \leq t} \overline{\U_\tau}
$
is nonempty and compact.

We now prove that $\omega_\AA(t)$ is attracting in the sense of Definition 1.4-(ii). Suppose not, then there exists $t \in \R$, a pullback-bounded family $\B$, $s_n \to -\infty$, $z_n \in \B(s_n)$ and $\delta >0$ such that
\begin{equation}
\label{proof15:1}
\inf_{z \in \omega_\AA(t)} \|S(t,s_n)z_n-z \|_{X_t} > \delta.
\end{equation}
Extract a subsequence $\{s_{n_k}\}$ from $\{s_n\}$ as follows: given $s_{n_1}, \ldots, s_{n_{k}}$, choose $s_{n_{k+1}}\leq s_{n_{k}}$ such that
$
S(s_{n_k},s_{n_{k+1}})\B(s_{n_{k+1}}) \subset \AA(s_{n_k}).
$
Now observe that
$$
S(t,s_{n_{k+1}})z_{n_{k+1}} = S(t,s_{n_k})S(s_{n_k},s_{n_{k+1}})z_{n_{k+1}} = S(t,s_{n_k}) w_{k},
$$
with $w_{k}=S(s_{n_k},s_{n_{k+1}})z_{n_{k+1}}  \in \AA(s_{n_{k}})$, i.e.
$$
S(t,s_{n_{k+1}})z_{n_{k+1}} \in \U_{s_{n_k}}.
$$
By  property (K.5) of the Kuratowski measure, the sequence $S(t,s_{n_{k+1}})z_{n_{k+1}}$ has an accumulation point in $\omega_\AA(t)$, which contradicts \eqref{proof15:1}. Therefore, the first step is complete.
\vskip1.5mm \noindent $\diamond$ \hskip1mm Invariance (in the sense of Definition 1.4(i)) of $\omega_\AA(t)$.
\\ \noindent Let $t \geq s$. We aim to prove that $\omega_\AA(t) = S(t,s)  \omega_\AA(s)$.

 We first deal with the inclusion $\omega_\AA(t) \supset S(t,s) \omega_\AA(s)$, which is easier.
Let $z \in \omega_\AA(s)$, then, there exist  sequences $s_n \to -\infty$, $z_n \in \AA(s_n)$, such that
\begin{equation} \label{proof15:2}
\|S(s, s_n)z_n-z \|_{X_s} \to 0, \qquad n \to \infty.
\end{equation}
Now, extract again a subsequence $\{s_{n_k}\}$  as follows: given $s_{n_1}, \ldots, s_{n_{k}}$, choose $s_{n_{k+1}}\leq s_{n_{k}}$ such that
$
S(s_{n_k},s_{n_{k+1}})\AA(s_{n_{k+1}}) \subset \AA(s_{n_k}).
$
Hence, we set $$
w=S(t,s)z\in  S(t,s) \omega_\AA(s),\quad w_k=S(s_{n_k},s_{n_{k+1}})z_{n_{k+1}} \in \AA(s_{n_k}).$$ Thus, we also have
\begin{align*}
\|S(t, s_{n_k})w_k-w \|_{X_t} &= \|S(t,s)S(s,s_{n_{k}})w_k-S(t,s)z\|_{X_t} \\& =  \|S(t,s)S(s,s_{n_{k+1}})z_{n_{k+1}}-S(t,s)z\|_{X_t} \to 0
\end{align*}
as $k \to \infty$, using \eqref{proof15:2} and the  continuity of $S(t,s)$. This establishes that $w=S(t,s)z \in \omega_\AA(t)$ and henceforth $\omega_\AA(t) \supset S(t,s) \omega_\AA(s)$ as claimed.

We now turn to the reverse inclusion $\omega_\AA(t) \subset S(t,s) \omega_\AA(s)$. Let $z \in \omega_\AA(t)$ be arbitrary, and choose sequences $s_n \to- \infty$, $z_n \in \AA(s_n)$ such that $s_n \leq s$ for all $n $ and
\begin{equation} \label{proof15:3}
\|S(t, s_n)z_n-z \|_{X_t} \to 0, \qquad n \to \infty.
\end{equation}
Using the attraction property (ii), proven in the first part,
$$
\lim_{n \to \infty} \inf_{w \in \omega_s (\AA)} \|S(s,s_n)z_n -w \|_{X_s} =0,
$$
so that there exists a sequence $w_n \in \omega_\AA(s)$ satisfying
$$
\lim_{n \to \infty} \|S(s,s_n)z_n -w_n \|_{X_s}=0
$$
By compactness of $\omega_\AA(s)$, we have  $w_n \to w \in \omega_\AA(s)$ in  $X_s$ up to a subsequence. This yields $S(s,s_n)z_n \to w$, in $X_s$, and, by continuity of $S(t,s)$, $S(t,s_n)z_n \to S(t,s)w$ (in  $X_t$) as well, i.e. $z=S(t,s) w$, which completes the proof of the second inclusion, and in turn, of Theorem \ref{existence}.
\end{proof}

\section{The oscillon equation with a general potential}

\subsection{Setting of the problem.}
In the following, $|\cdot|$ and $\l \cdot,\cdot\r$ denote
respectively the standard norm and scalar product on $L^2(0,1)$;
$A$ denotes  $-\Delta$ on $(0,1)$ with periodic boundary
conditions.

The symbols $c$ and $\Q$ will stand respectively for a generic
positive constant and a generic positive increasing continuous
function, which may be different in different occurrences; when an
index is added, (e.g.\ $c_0, \Q_0$), the positive constant (resp.\
function) is meant to be specific. Similarly, the symbols
$\Lambda,\Lambda_\imath$ will denote certain energy-like
functionals occurring in the proofs. Finally, in most of the
proofs we adopt the shorthand $\varpi(t):=\e^{-2Ht}.$

 \vskip2mm
 We study the \emph{oscillon
equation} in space dimension $n=1$ with periodic boundary
conditions
\begin{equation}
\label{SYS} \tag{P}
\begin{cases}
\displaystyle
\ptt u(t)+H  \pt u(t) + \e^{-2Ht} A u(t) + \varphi(u(t))=0, & t \geq s,\\ \\
u(s)= u_0 \in H^1_{\rm{per}}(0,1), \pt u(s) = v_0  \in L^2(0,1).
\end{cases}
\end{equation}
Here, (except at the end of Section 4) we consider a general nonlinear potential $V \in \C^2(\R)$,
such that $V(0)=0$, and $\varphi=V'$  satisfies the following
assumptions:
\begin{itemize}
\item[(H0)] $\varphi(0)=0$;
\item[(H1)] there exist $a_0,a_2> 0$, $a_1,a_3\geq0$, $q \geq 2$ such that
$$ a_0|y|^{q-2} -a_1 \leq \varphi'(y) \leq a_2 |y|^{q-2}+a_3.
$$
In addition, when $q=2$ (sublinear case), we assume $a_0 > a_1$.
\end{itemize}

Note that assumptions (H0)-(H1), as well as assumption \eqref{phisecond} in Theorem \ref{fract:thm},  are satisfied by any (non-constant) polynomial $V$ with $V(0)=0$  and positive leading coefficient. In the class of polynomials, only the above mentioned polynomials   satisfy  (H0)-(H1).

Also note that, since $V(0)=\varphi(0)=0$, two consecutive integrations
of (H1) yield
\begin{equation} \label{bound:V1}
\textstyle  \frac{a_0}{q(q-1)}|y|^{q} -\frac{a_1}{2} y^2 \leq V(y) \leq \textstyle  \frac{a_2}{q(q-1)}|y|^{q} +\frac{a_3}{2} y^2.
\end{equation}
Moreover, integrating by parts and using (H0)-(H1), we have, say  for $y \geq 0$,
\begin{align*}
 y\varphi(y) - V(y) & = \int_0^y \sigma \varphi'(\sigma) \, \d \sigma\\
& \geq \int_0^y \sigma (a_0 |\sigma|^{q-2} -a_1) \, \d \sigma\\
& \geq \frac{a_0}{q} |y|^{q} - \frac{ a_1}{ 2} y^2,
\end{align*}
and similarly for $ y< 0$, so that
\begin{equation} \label{fritto} y\varphi(y) \geq V(y) +\frac{a_0}{q} |y|^q - \frac{a_1}{2}y^2 \geq V(y) - c_0,
\end{equation}
for some $c_0\geq 0$ depending only on $a_0,a_1,q$. In particular,
we can take $c_0=0$ whenever $a_1=0$.
\begin{remark} \label{physical}
For the physically interesting potentials
$$V_{+}(y)= \frac{y^2}{2}+ \frac{y^4}{4}+\frac{y^6}{6}, \qquad V_{-}(y)= \frac{y^2}{2}- \frac{y^4}{4}+\frac{y^6}{6},$$
the nonlinearities read
$\varphi_{\pm}(y)= y\pm y^3+y^5$.
Noting that $3y^2 < 1 +4y^4$, it is immediate to observe that $\varphi_{+}$ fulfills both (H0) and (H1), for $q=6$ and e.g.\ $a_0=5,a_1=0,a_2=6,a_3=2$. The same observation yields that, for $\varphi_{-}$, (H0) and
(H1) are satisfied with $q=6, a_0=1,a_1=0,a_2=5,a_3=1.$
\end{remark}
We set $$ \V(u) =\int_0^1 V(u(x))\, \d x. $$ In view of \eqref{bound:V1},
$\V(u) $ is well defined for every $u \in L^q(0,1)$, and \begin{equation}
\label{bound:V} b_0(\| u\|_{L^{q}}^{q}+|u|^2)-b_1 \leq \V (u) \leq
b_2(\| u\|_{L^{q}}^{q}+|u|^2),
\end{equation}
with $b_0,b_2>0$ and $b_1\geq 0$ depending only on the
$a_\imath$ ($\imath=0,\ldots,3$) and $q$;
in particular, $b_1=0$ whenever $a_1=0$.

 Problem \eqref{SYS} is rewritten in our abstract framework as follows. The phase spaces are the Banach spaces
$$
X_t = H^1_{\rm{per}}(0,1) \times L^2(0,1) \quad \textrm{with norms} \quad \|(u,v) \|_{X_t}= \e^{-Ht}|A^{1/2}u | + \| u\|_{L^{q}} + |v |.
$$
For simplicity, we set  $X= X_0$. It is clear that the spaces $X_t$ are all the same as linear spaces;  given $z \in X$, its injection in $X_t$ will still be denoted as $z$.  Since for every $t \in \R$, $z \in X$
$$
\min\{\e^{-Ht},1\} \|z \|_{X} \leq \|z \|_{X_t} \leq \max\{\e^{-Ht},1\} \|z \|_{X}
$$
the norms $\|\cdot \|_{X_s}$, $\|\cdot \|_{X_t}$ are equivalent for any fixed $t,s \in \R$. Again, we stress how the equivalence constants blow up when $t \to \pm \infty$.

\begin{remark}  \label{equivalence}
For some of the proofs below, it will be convenient to use    the natural energy of the problem at time $t$
\begin{equation} \label{energy}
\E_{X_t}(u,v) =  \e^{-2Ht}|A^{1/2}u |^2 + \frac{2}{q}\|
u\|_{L^{q}}^{q} + |u|^2+ |v |^2
\end{equation}
in place of the $X_t$-norm.
It is easy to see that the energy $\E_{X_t}(\cdot)$ is equivalent to the norm $  \|\cdot \|_{X_t}$, in the following sense:
\begin{itemize}
\item[$\cdot$]  a family $\B=\{\B(t) \subset X_t\}$ is pullback-bounded
if and only if $$\displaystyle \sup_{s\in (-\infty, t]} \sup_{z
\in \B(s)}\E_{X_s}(z)< \infty \qquad \forall t \in \R;$$
\item[$\cdot$] a sequence $\{z_n\} \subset X_t$ converges to $z \in X_t$ if and only if $\E_{X_t}(z_n-z)$ converges to zero.
\end{itemize}
\end{remark}

\subsection{Well-posedness and dissipativity.}
Under the assumptions made in the preceeding section, we are able
to obtain a well-posedness result.
\begin{theorem}   Problem \eqref{SYS} generates a strongly continuous process
$S(t,s): X_s \to X_t$, satisfying the following dissipative
estimate: for $z  \in X$,
 \begin{equation}
\label{two0} \E_{X_t}(S(t,s)z)  \leq K_0( \E_{X_s}(z))
\e^{-\mu(t-s)} + K_1, \qquad \forall t\geq s,
\end{equation}
with $ K_1
=4{c_1}^{-1}(c_0+b_1)$. The  positive constants $c_1,K_0,\mu$   are
explicitly defined in the proof, and depend only on the physical
parameters $H$, $a_\imath$ and q.

Furthermore, the following additional continuous dependence
property holds: for every pair of initial conditions
$z^\imath  \in X$ $(\imath=1,2)$ with
$\E_{X_s}(z^\imath) \leq R$ and every $t \geq s$, we have
\begin{equation}
\label{continuous:dep} \E_{X_t}\big[S(t,s)z^1-S(t,s)z^2\big] \leq
\Q_1(R)\exp\big((t-s)+  \e^{H\sigma t  }-\e^{H\sigma s }
\big)\E_{X_s}\big[z^1-z^2\big] , \end{equation} where  $\sigma=
\frac q2-1$ and $\Q_1$ is specified in the proof. \label{wp:th}
\end{theorem}

\begin{proof}
We begin with the formal derivation of   estimate
\eqref{two0}. Hereafter, $(u(t),\pt u(t))$ denotes the solution to \eqref{SYS} with initial time $s \in \R$ and initial condition $z=(u_0,v_0) \in X, $ which we assume to be sufficiently regular.

A  multiplication of (P) by $\pt u$ entails
\begin{equation}
\label{one} \ddt \left[\varpi|A^{1/2}u |^2 + |\pt u |^2  + 2 \V(u)
\right] + 2H[\varpi|A^{1/2}u |^2 + |\pt u |^2 ] =0,\end{equation}
while multiplying (P) by $u$ and then using \eqref{fritto} yields
\begin{equation}
\label{two} \ddt  \left[ H| u|^2+ 2\l \pt u,  u  \r \right] +
2\varpi|A^{1/2}u |^2 - 2|\pt u |^2 = -2 \l\varphi(u),u \r \leq
-2\V(u) + 2c_0.
\end{equation}
For $0 <\nu \leq \frac H4$ to be determined later, we add
\eqref{one} to $\nu$-times \eqref{two}. Setting
\begin{align*}
&\Lambda_1 =  \varpi|A^{1/2}u |^2 + |\pt u |^2  + 2\V(u) + \nu( H
|u|^2 +2  \l \pt u, u \r), \\
& \Lambda_\star = (H+\nu) \varpi |A^{1/2} u |^2 + (H-3\nu)|\pt
u|^2 -\nu^2 H|u |^2-2\nu^2 \l\pt u, u \r,
\end{align*}
we obtain
\begin{equation}
\label{three} \ddt \Lambda_1 + \nu \Lambda_1 + \Lambda_\star  +
H[\varpi |A^{1/2} u |^2 + |\pt u|^2] \leq 2\nu c_0.
\end{equation}
Then, we take advantage of \eqref{bound:V} to get  the bound
\begin{equation}
\label{bound} c_1\E_{X_t}[(u(t),\pt u(t))] - 2b_1 \leq \Lambda_1(t) \leq
c_2\E_{X_t}[(u(t),\pt u(t))]
\end{equation}
with $c_1= \min\{qb_0,1\}/2$ and $c_2$  a positive constant
depending (increasingly) on $H$ and $b_2$. Hence, exploiting the
left-hand inequality of \eqref{bound}  to control $|u|^2$ from
above, we find
\begin{align*}
\Lambda_\star & \geq  (H-3\nu-\nu^2)|\pt u|^2 -\nu^2 (H+1)|u |^2\\
& \geq -\nu^2 (H+1)|u |^2  \geq -\nu^2(H+1)c_1^{-1}\Lambda_1 -
2\nu^2 b_1c_1^{-1} \\ & \geq -\frac\nu2 \Lambda_1-\nu b_1,
\end{align*}
provided we restrict ourselves to $\nu \leq 2\mu
=\min\{1,H/4,(4H+4)^{-1}c_1\}$. Therefore, writing \eqref{three}
for $\nu=2\mu$ and dropping the rightmost summand on the lhs
yields
\begin{equation}
\label{three:a}  \ddt \Lambda_1 + \mu \Lambda_1
\leq 2\mu (2c_0+b_1).
\end{equation}
Multiplying the above inequality by $\e^{\mu t}$ and integrating
between $s$ and $t$, we obtain
\begin{equation}
\label{four}  \Lambda_1(t) \leq  \Lambda_1(s) \e^{-\mu(t-s)} +
2(2c_0+b_1);
\end{equation}
an exploitation of \eqref{bound} then leads to
\begin{align}
  c_1\E_{X_t}[(u(t),\pt u(t))]  &\leq  \Lambda_1(t) + 2b_1 \nonumber \\ &
\leq  \Lambda_1(s) \e^{-\mu(t-s)} + 4(c_0+b_1)  \nonumber \\ &
 \leq   c_2\E_{X_s}( z)  \e^{-\mu(t-s)}   + 4(c_0+b_1) \label{four:a}
\end{align}
which is \eqref{two0}, with $K_0=c_2/c_1, K_1
=4c_1^{-1}(c_0+b_1).$ Note that, like $b_1$ and $c_0$, $K_1=0$
when $a_1=0$ in (H1).

Having at our disposal \eqref{two0}, global existence of (weak)
solutions $(u(t),\pt u(t))$ to Problem \eqref{SYS}  is obtained by
means of a standard Galerkin scheme. The solutions we obtain in
this way satisfy, on any interval $(s,t)$, $-\infty<s<t<+ \infty,$
$$
u \in L^\infty\big(s,t; H^1_{\mathrm{per}}(0,1)\big) \cap
L^q\big(s,t;L^q(0,1)\big), \quad \pt u \in L^\infty\big(s,t;
L^2_{\mathrm{per}}(0,1)\big).
$$
Replacing $L^\infty$ on $(s,t)$ with continuity on $[s,t]$
requires some additional work, as explained in \cite[Section
II.4]{TEM}.

 \vskip1.5mm
 Uniqueness of solutions, and
therefore generation of the process $S(t,s)$ will then follow once
the continuous dependence estimate \eqref{continuous:dep} is
established.

To prove \eqref{continuous:dep}, for $\imath=1,2$, let
$z^\imath=(u_0^\imath,v_0^\imath) \in X$ with $\|z^\imath\|_{X_s}
\leq R$. Accordingly, call $(u^\imath(t),\pt
u^\imath(t))$ the solution corresponding to initial datum $z^\imath$, prescribed at time $s \in \R$. Preliminarily, we recall that the
dissipative estimate  \eqref{two0}   can be rewritten as
\begin{equation}
\label{dependence:1} \E_{X_t}[(u^\imath(t),\pt u^\imath(t))]  \leq K_0 R + K_1:=\Q(R), \qquad \forall
t\geq s.
\end{equation}

Then, we observe that the difference $$\bar z(t) = (u^1(t),\pt u^1(t))-
(u^2(t),\pt u^2(t)) = (\bar u(t),\pt \bar u(t))$$ fulfills the Cauchy
problem on $(s,+\infty)$
$$
\begin{cases}
\displaystyle \ptt \bar u+H  \pt \bar u + \varpi A \bar u + \bar u|\bar
u|^{q-2}   +\bar  u=   \bar u + \varphi(u^2)-\varphi(u^1) +  \bar u|\bar
u|^{q-2},
 \\
\bar z(s) = z^1-z^2.
\end{cases}
$$
Assuming $(\bar u ,\pt \bar u )$ sufficiently smooth, we multiply the above equation by $\pt \bar u$ and obtain the
differential inequality
\begin{equation} \label{dependence:2}
\ddt \E_{X_t}(\bar z) \leq 2 \l\bar u + \varphi(u^2)-\varphi(u^1)  + \bar u|\bar
u|^{q-2}, \pt \bar u \r.
\end{equation}
The first term in the rhs is easily estimated by $
2|\bar u||\pt \bar u|$. Regarding the second, we exploit the Agmon inequality to obtain the bound
$$
\|u^\imath\|_{L^\infty} \leq c ( |u^{\imath}|  + |u^{\imath}|^{\frac12}|A^{1/2}u^{\imath}|^{\frac12}), \qquad \imath=1,2,
$$
so that,
using (H1),   we estimate
\begin{align*}
2 \l \varphi(u^2)-\varphi(u^1), \pt \bar u \r & =-2 \l \varphi'(\xi)\bar u, \pt \bar u \r \\
& \leq c \big(1+\|u^1\|_{L^{\infty}}^{q-2} + \|u^2\|_{L^{\infty}}^{q-2}\big)|\bar u| |\pt \bar u|\\
& \leq c \Big(1+\sum_{\imath=1,2}( |u^{\imath}|  + |u^{\imath}|^{\frac12}|A^{1/2}u^{\imath}|^{\frac12})^{q-2}  \Big)|\bar u| |\pt \bar u|.
\end{align*}
Treating $|\bar u|^{q-2}$ as done above for $\varphi'$ yields the similar control
$$
2 \l \bar u |\bar u|^{q-2}, \pt \bar u \r \leq    c\Big(1+\sum_{\imath=1,2}( |u^{\imath}|  + |u^{\imath}|^{\frac12}|A^{1/2}u^{\imath}|^{\frac12})^{q-2}  \Big)|\bar u| |\pt \bar u|.
$$
Recalling \eqref{dependence:1} and the obvious bounds
$$
|u^\imath(t)| \leq \E_{X_t} (u^\imath(t),\pt u^\imath(t))^{1/2}, \qquad | A^{1/2}u^{\imath}(t)|^{\frac12} \leq \e^{\frac{Ht}{2}} \E_{X_t} (u^\imath(t),\pt u^\imath(t))^{1/4},
$$
we have the estimate
$$
\sum_{\imath=1,2}( |u^{\imath}(t)|  + |u^{\imath}(t)|^{\frac12}|A^{1/2}u^{\imath}(t)|^{\frac12})^{q-2} \leq \Q(R)^\sigma(1+\e^{ \sigma Ht }),
$$
where we have set $\sigma=\frac q2 -1$,
so that the controls on the rhs of \eqref{dependence:2} may be summarized as
\begin{equation} \label{dependence:2a}
\ddt \E_{X_t}(\bar z(t)) \leq c\Q(R)^{\sigma}(1+\e^{\sigma Ht})\E_{X_t}(\bar
z(t)).
\end{equation}
We then apply Gronwall's lemma on $(s,t)$ to get
\begin{align*}
&\E_{X_t}(\bar z(t)) \leq   \Q_1(R)\exp\big((t-s)+ \e^{H\sigma t
}-\e^{H\sigma  s } \big) \E_{X_s}(z^1-z^2),
\end{align*}
where $\Q_1(R)=c\exp\big(\Q(R)^{\sigma}\big)$, as claimed in \eqref{continuous:dep}, so that the proof is complete.
\end{proof}

\begin{remark} Note that, in the proof of Theorem \ref{wp:th},
we have assumed $(u,\pt u)$ and $(\bar u, \pt \bar u)$ to be
smooth enough to derive the estimates \eqref{one} and
\eqref{dependence:2}, which led us to \eqref{two0} and
\eqref{continuous:dep}. These energy inequalities can be made rigorous by
deriving them for the corresponding Galerkin approximation and
then passing to the (lower) limit.
Energy equalities (which were not needed here) can be obtained by
the regularization procedure described in \cite[Lemma
II.4.1]{TEM}.

Finally, the derivation of the a priori estimate
\eqref{three}, based on adding eq.\ \eqref{EQ-INTRO} multiplied by
$u$ to $\nu$-times eq. \eqref{EQ-INTRO} multiplied by $\pt u$ is
reminiscent of \cite{GT,TEM}, where we multiply \eqref{EQ-INTRO}
by $u + \nu \pt u $. However, here we work with the original
variables $u,\pt u$ in place of $u, u + \nu \pt u$ as in
\cite{GT,TEM}.

 \end{remark}

 Using  the dissipative estimate  \eqref{two0}, we
 explicitly construct an absorber for the process generated by
 \eqref{SYS}.
\begin{theorem}There exists $R_\AA=R_\AA(H,a_\imath)>0$ such that
the family
\begin{equation} \label{radius}
\AA=\big\{\AA(t)= \{z \in X_t : \E_{X_t}(z) \leq
R_{\AA} \}\big\}
\end{equation} is an absorber for the process
$S(\cdot,\cdot).$ \label{dissipativity:th}
\end{theorem}
\begin{proof}  Let $\B$ be a pullback-bounded family and, for $t \in \R$, let
$$
R(t)=\sup_{s \in (-\infty,t]}\E_{X_s} [\B(s)].
$$
Observe that Remark~\ref{equivalence} guarantees $R(t) < \infty$ for every $t \in \R$.
Estimate \eqref{two0} then reads
$$
\E_{X_t}(S(t,s)z) \leq K_0 \E_{X_s}(z) \e^{-\mu(t-s)} + K_1 \leq
K_0 R(t) \e^{-\mu(t-s)} +   K_1 \leq 1 +2K_1
$$
for every $z \in \B(s)$, provided that
$$
 s \leq t_0 =t-\max\big\{0,\mu^{-1}\textstyle\log \frac{K_0 R(t)}{1+K_1} \big\}.
$$
Taking the supremum over $z \in \B(s)$ , we obtain
$$
\E_{X_t}[S(t,s)\B(s)]  \leq 1+2K_1, \qquad \forall s \leq t_0,
$$
which, setting $R_\AA=1+2K_1 $, reads exactly $S(t,s)\B(s) \subset
\AA(t)$ whenever $s \leq t_0$.
\end{proof}

\section{The Global  Pullback Attractor.}
We now devote ourselves to the construction of a global attractor (in the sense specified in Section \ref{sect:attractors}) for the oscillon equation. Existence of an attractor, uniqueness, and regularity property are specified in the main theorem below.
\begin{theorem} \label{osc:attractor}  The family $\A(t)=\omega_\AA(t)$ is the unique (in the sense of Remark~\ref{uniqueness}) global attractor
of the process $S(\cdot, \cdot)$ generated by (P). Moreover, introducing the family of Banach spaces ($t \in \R$)
$$
Y_t = H^2_{\rm{per}}(0,1) \times H^1(0,1),\qquad \|(u,v) \|_{Y_t}= \e^{-Ht}|Au | + \|u\|_{L^q} +
|A^{1/2}v|+|v |
$$
we have
$$
\A(t)\subset Y_t, \quad\|\A(t)\|^2_{Y_t} \leq h(t), \qquad \forall t \in \R,
$$
where $h(t)$ is a continuous increasing function of $t$, depending
on $R_\AA$, and is defined below.
\end{theorem}
\begin{remark} \label{improvebound} Observe that the injections $\I_t: Y_t \to X_t$ satisfy $$\|\I_t\|_{\LL(Y_t,X_t)} \leq \max\{1,\lambda_1^{-1}\}$$ (here, $\lambda_1$ stands for the Poincar\'e constant on $(0,1)$).
Hence, the estimate of Theorem \ref{osc:attractor} implicitly asserts that $\A$ is a pullback-bounded family in $X_t$. Then, the invariance of $\A$ implies that $\A(t) \subset \AA(t)$ for every $t \in \R$; more explicitly,
\begin{equation} \label{boundX}
\sup_{(u,v) \in \A(t)} \left[e^{-2Ht}|A^{1/2}u |^2 +\textstyle \| u\|_{L^{q}}^{q} + |v |^2\right] \leq R_\AA, \qquad \forall t \in \R.
\end{equation}
\end{remark}
\begin{remark} \label{Hirrelevant}  Note that the Hamiltonian structure of the oscillon system is irrelevant to the existence of a global
attractor. Such structure exists for any potential $V$ or for $V=0$. But the proof of Theorem \ref{osc:attractor}
depends on a potential $V$ with the property \eqref{fritto}, through Lemma \ref{lemma:P} below.
\end{remark}
We turn to the proof of Theorem \ref{osc:attractor}. We will work
throughout  with
$$
z=(u_0,v_0) \in \AA(s);
$$
until the end of the section, the generic constants $c>0$
appearing depend only on $R_\AA$, whose dependence   on the
physical parameters of the problem has been specified earlier.
Hence, the estimate \eqref{two0} now reads
\begin{equation}
\label{palmiro} \E_{X_t}(S(t,s) z)  \leq  K_0 R_\AA + K_1 :=c,
\qquad \forall s \in \R, t     \geq s.
\end{equation}
Now, with the aim of using Corollary \ref{cor:attractor}, we perform a suitable decomposition of the solution of Problem (P). We set
$$
(u(t),\pt u(t)) = S(t,s)z = P_{z}(t,s) + N_z(t,s) = (p(t),\pt p(t)) + (n(t),\pt n(t)),
$$
where
\begin{equation}
\label{SYS-P}
\begin{cases}
\displaystyle
\ptt p(t)+H  \pt p(t) + \e^{-2Ht} A p(t) + \varphi_\star(p(t)),=0, & t \geq s,\\ \\
p(s)= u_0 , \pt p(s) = v_0,
\end{cases}
\end{equation}

\begin{equation}
\label{SYS-Q}
\begin{cases}
\displaystyle
\ptt n(t)+H  \pt n(t) + \e^{-2Ht} A n(t) = \varphi_\star(p(t))-\varphi(u(t)), & t \geq s,\\ \\
n(s)= 0, \pt n(s) = 0.
\end{cases}
\end{equation}
and $ \varphi_\star (y) = y +|y|^{q-2}y.$
\begin{lemma} \label{lemma:P}
There exists  $K_2,\mu_1 >0$ so that
$$
\E_{X_t}(P_z (t,s)) \leq K_2R_\AA \e^{-\mu_1 (t-s)} \qquad \forall
t \in \R, s \leq t.
$$
\end{lemma}
\begin{proof}
We peruse the proof of  \eqref{two0}, Theorem \ref{wp:th},
replacing $\varphi $ with  $\varphi_\star$. In this case (H1)
holds with $a_1=0,$ and the corresponding potential $V_\star(y) =
y^2/2 +|y|^{q}/q$ satisfies \eqref{fritto} with $c_0=0$, and
\eqref{bound:V} with (e.g.) $b_0=1/q,b_2=1$, and $b_1=0$. Hence
$K_1=0$ in \eqref{two0}, which is exactly the claimed estimate.
Observe that the constants $\mu_1$ and $K_2$ can be explicitly
computed.
\end{proof}
\begin{lemma}There exists a continuous positive increasing function $h$  such that \label{lemma:Q}
\begin{equation} \label{lemma:Q-0}
\sup_{s\in (-\infty,t]} \sup_{z \in \AA(s)} \|N_z(t,s) \|_{Y_t}^2 \leq h(t) \qquad \forall t \in \R.
\end{equation}
\end{lemma}
\begin{proof} We first observe that $n(t)=u(t)-p(t)$, so that, using \eqref{palmiro} and
Lemma~\ref{lemma:P},
\begin{equation} \label{lemma:Q-1}
\|n(t) \|^q_{L^q} + \e^{-2Ht} |A^{1/2} n(t) |^2+ |\pt n(t)|^2 \leq
c.
\end{equation}
Therefore, we are only left to control the higher-order seminorms
appearing in $\|N_z(t,s) \|_{Y_t}$. To this aim, assuming $(n,\pt n)$ is sufficiently regular,  we multiply the
equation \eqref{SYS-Q}
 by $A\pt n$, obtaining for the functional
$$
\Lambda_2 = \varpi|An |^2 + |A^{1/2} \pt n|^2
$$
the differential equation
$$
\ddt \Lambda_2 + 2H\Lambda_2    = -2\l\varphi (u)- \varphi_\star(p), A\pt n  \r.
$$
Using (H1) with the Agmon inequality, and taking advantage of
\eqref{palmiro} in the last inequality yields
\begin{align*}
-2\l\varphi (u), A\pt n  \r & = -2\l \varphi'(u)\partial_x u, A^{1/2} \pt n\r
\\ &\leq c(1+|u |^{q-2}|A^{1/2} u |^{q-2})|A^{1/2} u |^{2} +  H|A^{1/2} \pt n |^2 \\
& \leq c (1+\|u \|_{L^q}^{q-2})|A^{1/2} u |^{q} +  H|A^{1/2} \pt n |^2 + c \\
& \leq c \varpi^{-q/2} + H|A^{1/2} \pt n |^2 + c.
\end{align*} Similarly,
\begin{align*}
2\l \varphi_\star(p) ,A \pt n  \r &  \leq c(|A^{1/2}p|^2
+|p|^{q-2}|A^{1/2} p |^q) +  H  |A^{1/2} \pt n |^2
\\ & \leq  c(|A^{1/2}p|^2+\|p\|_{q}^{q-2}|A^{1/2} p  |^q) +  H  |A^{1/2} \pt n |^2 \\
& \leq c\varpi^{-q/2} +  H |A^{1/2} \pt n |^2
\end{align*}
 where Lemma~\ref{lemma:P} is used in the last inequality.
Summarizing, for a fixed $t>s$, we arrive at the differential
inequality
$$
\frac{\d}{\d \tau } \Lambda_2(\tau) + H \Lambda_2 (\tau)   \leq
c(1+\e^{qH\tau} ) \leq c(1+\e^{qHt}), \qquad \forall \tau \in
[s,t].
$$
 An application of the Gronwall lemma then yields
\begin{equation} \label{lemma:Q-2}
\varpi\|An(t) \|^2 + \|A^{1/2} \pt n(t)\|^2 = \Lambda_2(t) \leq
c(1+\e^{qHt}),
\end{equation}
due to the fact that $\Lambda_2(s)=0$. Combining \eqref{lemma:Q-1}-\eqref{lemma:Q-2}, we finally arrive at
$$
\|N_z(t,s) \|_{Y_t}^2 \leq c(1+\e^{qHt}) := h(t),
$$
for all $t \in \R, s\leq t $ and $z \in \AA(s)$, which is what we
looked for.
\end{proof}
We can now complete the proof of the main theorem.
\begin{proof}[Proof of Theorem \ref{osc:attractor}]
First of all, we remark again  that each $Y_t$ is compactly
embedded in $X_t$ and the injections $\I_t: Y_t \to X_t$ satisfy
$\|\I_t\|_{\LL(Y_t,X_t)} \leq \max\{1,\lambda_1^{-1}\}$. Finally,
each $Y_t$ is a reflexive Banach space, so that closed balls of
$Y_t$ are closed in $X_t$. These considerations ensure that we are
in   position to apply Corollaries \ref{cor:attractor} and
\ref{cor:regularity}.

 Setting $$
P(t,s) = \bigcup_{z \in \AA(s)} P_z(t,s), \qquad N(t,s) =
\bigcup_{z \in \AA(s)} N_z(t,s),
$$
we have $S(t,s) \AA(s) \subset P(t,s)+ N(t,s)$. Lemma \ref{lemma:P} grants
$$
\lim_{s\to -\infty} \|P(t,s) \|_{X_t} = \lim_{s \to -\infty}
\sup_{z \in \AA(s)} \|P_z(t,s) \|_{X_t} =0,
$$
(observe that convergence to zero in $\|\cdot\|_{X_t}$ is equivalent to convergence to zero of $\E_{X_t}(\cdot)$, as stated in Remark \ref{equivalence}),
while Lemma \ref{lemma:Q} establishes that $$ \sup_{s\leq t}\|
N(t,s) \|_{Y_t}^2 \leq h(t).
$$ Therefore, $N(t,s) $ is compact in $X_t$, for every $t \in \R, s \leq t$. Applying Corollary \ref{cor:attractor}, we achieve the existence of  the global attractor $\A(t)=\omega_\AA(t)$. Finally, estimate \eqref{lemma:Q-0} and   Corollary \ref{cor:regularity} yields $\|\A(t)\|^2_{Y_t} \leq h(t)$, and the uniqueness of $\A$ in the sense of Remark \ref{uniqueness}.
\end{proof}

\section{The physical potentials.}
In this short section we apply the general results obtained in Sect.\ 3 and 4 to the  potentials discussed in the physics literature
$$
V_{\pm}(y) = \frac12 y^2 \pm \frac14 y^4 + \frac16 y^6
$$
(see also Remark \ref{physical}). We also make some remarks concerning the forward asymptotic behavior (i.e.\ we fix an initial time $s \in \R$ and let $t \to +\infty$), for the whole class of potentials satisfying our assumptions (in particular, $V_\pm$).

\subsection{Pullback exponential decay of \eqref{EQ-INTRO} with $V_\pm$}
Apart from the regularity result of Theorem \ref{osc:attractor}, we do not dwell on the structure of the global pullback attractor for a general potential satisfying our assumptions (H0),(H1). However,  for both     potentials
$V_{+},V_{-}$ described in Remark \ref{physical},
the pullback attractor $\A $ is reduced to zero, i.e.\ $\A(t)=\{0\}$ for every $t \in \R$.

Indeed, both $V_{-}$ and $V_{+}$ comply with assumptions (H0)-(H1), with (in particular) $a_1=0$.
As a result, we can take $c_0=0$ in \eqref{fritto}, and
\eqref{bound:V} holds with
$b_1=0$. We then read in    Theorem \ref{wp:th} that
estimate \eqref{two0} holds with $K_1=0$ and can therefore be rewritten as
\begin{equation}
\e^{-2Ht}|A^{1/2}u(t) |^2 + \| u(t)\|_{L^{q}}^{q} + |\pt u(t) |^2 \leq cR(t)\e^{-\mu(t-s)},
 \qquad \forall t \geq s
\label{bingo}
\end{equation}
whenever $z \in \B(s)$, with  $$R(t)=\sup_{s\leq t} \sup_{z \in \B(s)} \E_{X_s}(z)
< \infty.$$ Hence, the rhs of \eqref{bingo}, and thus the lhs, go to zero as $s \to -\infty$. We conclude that the potentials $V_+$, $V_{-}$ (like any admissible
potential fulfilling (H1) with $a_1=0$) produce (pullback) exponential decay of the solutions
originating from pullback-bounded inital data. In other
words, the family $\big\{\A(t)=\{0\}\big\}_{ t \in \R}$ is the (unique, in the sense of Remark \ref{uniqueness}) time-dependent
global attractor for the process generated by \eqref{SYS}.

\begin{remark}
It is observed in the physical literature (see also the caption to
Figure \ref{figoscillon}) that the term $\frac{y^6}{6}$ in the
potentials $V_{\pm}$ does not have a deep physical meaning.
Accordingly, we consider the modified potentials
$$
V_{\alpha \pm} (y) =   \frac12 y^2 \pm \frac14 y^4 +
\frac{\alpha}{6} y^6
$$
for $\frac14< \alpha <\frac{9}{20}$.  The potentials $V_{\alpha
-}$ still have a unique local and  global minimum at 0, but fail to be convex, and henceforth  the constant
$a_1$ appearing in (H1) must be taken strictly positive. On the contrary, $a_1$ can  be chosen  to be zero for
$V_{\alpha +}$. Therefore, we still have pullback exponential decay for
$V_{\alpha +}$, whereas the result is not known for $V_{\alpha
-}$; that is, in this case we cannot conclude that the pullback
attractor is trivial.

The range $0 <\alpha <\frac14$ for $V_{\alpha-}$, where two nontrivial (i.e.\ negative) global minima  appear, is less relevant for the physical problem under consideration: however, assumptions (H0)-(H1) (with, necessarily, $a_1 >0$) are still valid.

\end{remark}

\subsection{Forward convergence} We now consider the oscillon equation between times $s$ and $t$, $s\in \R$ fixed, $t>s$ with the aim of letting $t \to +\infty.$
We assume that the initial data $z=(u_0,v_0)$ satisfies $\E_{X_s}(z) \leq R$. We then infer from \eqref{two0} that
\begin{equation}
\e^{-2Ht}|A^{1/2}u(t) |^2 + \| u(t)\|_{L^{q}}^{q} + |\pt u(t) |^2 \leq K_0R\e^{-\mu(t-s)}
 +K_1\qquad \forall t \geq s
\label{bingo2},
\end{equation}
where $0 <\mu <H$, $K_0$ and $K_1=c(c_0+b_1)$ have been computed in Theorem \ref{wp:th}.

As mentioned above, in the case of the potentials $V_-$ and $V_+$,
$K_1=0$.  Consequently, we observe that,  the solution $(u(t), \pt
u(t)) = S(t,s) z$ decays exponentially to zero in the  norm  $L^q
\times L^2$ as $t \to +\infty$ (forward convergence). Due to the
fact that the bound on the norm of the initial data $z$ depends on
the initial time $s$, we cannot conclude that the above mentioned
convergence is uniform in $s$. Therefore, we cannot conclude that
$\{0\}$ is the (weak, i.e.\ $L^q \times L^2$) uniform attractor
for $S(t,s)$ in the sense of Babin and Vishik \cite{BV,CV1}.
Furthermore, since $\mu$ must be less than $H$ in \eqref{two0},
nothing can be said about the behavior of $|A^{1/2}u(t) |$ as $t
\to +\infty$ (while the pullback approach grants a bound at every
fixed time $t$, see above). \emph{Hence, this analysis leaves also
open the question of the behavior of $|A^{1/2}u(t) |^2$ as $t \to
+ \infty$.}

With a further analysis, we can extend the  result of decay of $\pt u$  described above to the whole class of potentials satisfying assumptions (H0)-(H1), though the rate of decay will be only polynomial in time, instead of exponential as in the case of potentials $V_{\pm}$. The proof of the following proposition relies on an adaptation of the argument in \cite[Lemma 2.7]{BP}.
\begin{proposition} \label{prop}
Let $s  \in \R$ be fixed,  $z=(u_0,v_0) \in X$ such that
 $
 \E_{X_s}(z) \leq R.
 $
 We have the estimate
\begin{equation} \label{decayfw}
\|\pt u(t)\|^2  \leq \Q_2(R) \frac{1}{1+(t-s)},
\end{equation}
for every $t \geq s$, where $\Q_2$ is specified below. \end{proposition}
\begin{proof}   Setting $$\E(t)= \e^{-2Ht}|A^{1/2}u(t) |^2 + |\pt u(t) |^2, \quad \Phi(t)=\E(t)+2\V(u(t)),$$ we rewrite \eqref{one} as
\begin{equation}
\label{onestar}
\ddt \Phi + 2H \E =0.
\end{equation}
Also note, as a consequence of \eqref{bound:V}, that
\begin{equation}
\label{tre}
\Phi(t) \geq -2b_1, \qquad \forall t \geq s.
\end{equation}
A further consequence of \eqref{bound:V} is that
\begin{equation}
\label{4} \Phi(s) = \e^{-2Hs}|A^{1/2}u_0 |^2 + |v_0 |^2  + 2
\V(u_0) \leq c(1+R+R^{2/q}):=\Q(R)
\end{equation}
whenever $  \E_{X_s}(z) \leq R $. \vskip2mm \noindent Let $\delta
>0$ be given and set  $$ t_\delta =s+ \frac{\Q(R)+2b_1}{2H\delta}.
$$ We will show that
\begin{equation}
\label{one12}
\ddt \Phi(t) \geq -2H\delta
\end{equation}
for every $t \geq t_\delta.$ We first show that there exists $t_0
\in [s,t_\delta]$ such that \eqref{one12} holds for $t=t_0$.
Suppose it were not so, then $\ddt \Phi(t)< -2H\delta $ on
$[s,t_\delta]$, which would imply, by integration,
$$
\Phi(t_\delta) <\Phi(s) -2H\delta (t_\delta - s) \leq  \Q(R) -
2H\delta \frac{\Q(R)+2b_1}{2H\delta} = -2b_1
$$
which contradicts \eqref{tre}.
Now, define
$$
t^\star = \sup\big\{ \tau \geq t_0 : \textrm{ \eqref{one12} holds } \forall t \in [t_0,\tau]\big\}.
$$
We show that $t^\star = + \infty$. Indeed, if it were not so we could  pick $t_n \downarrow t^\star$, for which $\ddt \Phi(t_n) < -2H\delta$. Consequently, for any given $\eps>0$, we can find  another sequence $\tau_n$, with $t^\star< \tau_n < t_n $, such that
$$
\Phi(\tau_n) > \Phi(t^\star) - \eps,
\qquad
\Phi(\tau_n) < \Phi(t_n).$$
This implies that $
 \Phi(t^\star) < \Phi(t_n) + \eps$, and henceforth, $\eps$ being arbitrary, $\Phi(t_n) -\Phi(t^\star) \geq 0$. In turn, this leads to $\ddt \Phi(t^\star)\geq 0$. By continuity of $\ddt \Phi$ (which is a consequence of \eqref{onestar} and the continuity properties of the solutions), $\ddt \Phi(t) \geq 0$ in a right neighborhood of $t^\star$, which contradicts maximality.

Hence, \eqref{one12} holds for all $t \geq t_\delta$:  inserting this in \eqref{onestar} yields immediately
$$
|\pt u(t)  |^2 \leq  \E(t) \leq \delta, \qquad \forall t \geq t_\delta.
$$
To get  \eqref{decayfw} for a given $t>s+1$,   it is sufficient to
choose $\delta =  \frac{\Q(R)+2b_1}{2H(t-s)}$, so that $t_\delta
=t$, and
$$
|\pt u(t)  |^2 \leq  \frac{\Q(R)+2b_1}{2H(t-s)}.
$$
On the other hand, if $0\leq t-s < 1$, integrating \eqref{onestar} gives
$$
|\pt u(t)|^2 \leq \Phi(t)+2b_1 \leq \Phi(s)+2b_1 \leq \Q(R) +
2b_1.
$$
Combining the last two inequalities yields \eqref{decayfw} with
$\Q_2(R)=(\Q(R)+2b_1)\max\{1,H^{-1}\}$.
 \end{proof}

\section{Finite-Dimensionality of the Global Attractor.}
We conclude the paper with a study of the fractal dimension of the
pullback attractor $\A$ of system (P) constructed above.

\subsection*{Fractal dimension}
 For a Banach space $W$, let $B_W$ denote the closed unit ball of $W$. For $\eps>0$, we call $\eps$-ball centered at $x \in W$ the set $B_W(\eps,x)=x+\eps B_W.$

If $K \subset W$ is compact, we use $ \mathcal N_\eps (K,W) $ to
indicate the minimum number of $\eps$-balls of $W$ which cover
$K$, and we define the \emph{fractal}
\emph{dimension} of $K$ as
$$
\dim_{W} K = \limsup_{\eps \to 0^+} \frac{\log \mathcal
N_\eps(K,W)}{\log \textstyle \frac1\eps}
$$
For more details on the fractal dimension (also known as the Minkowski or \emph{box-counting} dimension), we refer the reader to e.g.\ \cite{MAN,SCH} (see also \cite{TEM}).

\begin{remark}
\label{banspaceiso}Directly from the definition, it follows that  Banach space isomorphisms preserve
fractal dimension: more precisely, if $W,\widetilde W$ are two
Banach spaces, $K \subset W$ is compact and $J:W \to \widetilde W$
is a Banach space isomorphism, then $ \dim_{\widetilde W} J(K) =
\dim_W K. $
\end{remark}

The following abstract lemma (an adaptation of the generalized squeezing property  method \cite{FGMZ,GGMP} to our  framework)   is the main
technical tool we need in order to establish a bound on the
fractal dimension of $\A(t)$.

\begin{lemma}For $k \in \mathbb{N}$, let  $W_k,Z_k$ be two families of Banach spaces
satisfying \label{lemma:frac}
\begin{itemize}
\item[(i)] $Z_k \Subset W_k$;
\item[(ii)] for each $\eps>0$,  $\kappa_\eps= \sup_{k \geq 0}\mathcal N_{\eps}\big(B_{Z_{k}}(1,0),W_{k}\big) < \infty$.
\end{itemize}
Let $\B=\{\B_k \subset W_k\}_k$ be a family of sets with maps $U^k:
\B_k \to \B_{k-1}$, $k \geq 1$, fulfilling
\begin{itemize}
\item[(iii)] $\B_k$ is compact in $W_k$;
\item[(iv)] $ \sup_{k \geq 0}\|\B_k\|_{W_k} = Q_1 < \infty$;
\item[(v)] $U^k(\B_k) = \B_{k-1}$;
\item[(vi)] there exists a decomposition $U^k (z) = P^k (z) + N^k(z)$ and constants $\varrho<\frac14,Q_2>0$ such that
\begin{equation} \label{decay}
\|P^k( z^1)-P^k( z^2)  \|_{W_{k-1}} \leq \varrho \|z^1 -z^2
\|_{W_{k}},
\end{equation}
and
\begin{equation} \label{smoothing}
\|N^k( z^1)-N^k( z^2)  \|_{Z_{k-1}} \leq Q_2 \|z^1 -z^2
\|_{W_{k}},
\end{equation}
for every $z^1,z^2 \in \B_k$.
\end{itemize}
Then,
$$
\dim_{W_0} \B_0 \leq \frac{\log_2\kappa_{\varrho Q_2^{-1}}}{\log_2 \textstyle
\frac{1}{4\varrho}}.
$$
\end{lemma}
\begin{proof}
Let $0<\eps<Q_2$ and $k \geq 1$ be fixed. By compactness,
$\B_k$ can be covered by a finite number $\eta_\eps=\eta_\eps(\B_k,W_k)$ of
$\eps$-balls $\{B_{W_k}(\eps,z^\imath)\}_{\imath=1}^{\eta_\eps}$
with $z^\imath \in \B_k$. Then, from  \eqref{smoothing}, we learn
that for any fixed $z \in \B_k$ there exists $z^\imath$ such that
$$
 \|N^k (z)-N^k(z^\imath)  \|_{Z_{k-1}} \leq \eps Q_2  \implies N^k (z) \in B_{Z_{k-1}}(\eps Q_2,N^k(z^\imath)),
$$
so that $N^k(\B_k)$ is covered by the set of balls
$\big\{B_{Z_{k-1}}\big(\eps Q_2, N^k
(z^\imath)\big)\big\}_{\imath=1}^{\eta_\eps}$. Now, we cover each
ball in this set by a finite number of $\varrho\eps$-balls of
$W_{k-1}$. The minimum number of such balls is given by
$$
\mathcal N_{\rho\eps}\big(B_{Z_{k-1}}(\eps Q_2,0),W_{k-1}\big) =
\mathcal N_{\varrho Q_2^{-1}}\big(B_{Z_{k-1}}(1,0),W_{k-1}\big) =
\kappa_{\varrho Q_2^{-1}}:=\kappa.
$$
Hence, there exists a collection $\{B_{W_{k-1}}(\varrho\eps,
y^{\imath,\jmath})\}_{\imath,\jmath=1}^{\eta_\eps,\kappa}$, with
$y^{\imath,\jmath}\in W_{k-1}$, covering $N^k(\B_k)$. This means that, for any fixed $z \in B_{W_k}(\eps,z^{\imath})$, there exist $\jmath $, $y^{\imath,\jmath}$ such that
$$
\|U^k(z) -(y^{\imath,\jmath}+P^k(z^\imath)) \|_{W_{k-1}} \leq
\|N^k(z) -y^{\imath,\jmath}\|_{W_{k-1}}+\| P^k(z)-
P^k(z^\imath)\|_{W_{k-1}} \leq 2\varrho  \eps,
$$
i.e.\ if $\beta:=4\varrho <1$,  $U^k (z) \in
B_{W_{k-1}}(\frac{\beta\eps}{2},y^{\imath,\jmath}+P^k(z^\imath))$,
so that
$$
U^k \B_k \subset \bigcup_{\imath=1}^{\eta_\eps}
\bigcup_{\jmath=1}^{\kappa} B^{\imath,\jmath}, \qquad  B^{\imath,\jmath}=B_{W_{k-1}}\big(\textstyle
\frac{\beta\eps}{2},y^{\imath,\jmath}+P^k(z^\imath)\big).
$$

It might happen that some  $y^{\imath,\jmath}+P^k(z^\imath)$  does not belong to $\B_{k-1}$. In this case, choose $\tilde y \in \B_{k-1} \cap B^{\imath,\jmath}$ and replace $B^{\imath,\jmath}$ with $B_{W_{k-1}}\big(\textstyle
\beta\eps,\tilde y \big)$. Therefore, a system of $\kappa \eta_\eps $ $\beta   \eps$-balls of $W_{k-1}$ is sufficient to cover $\B_{k-1}$, and we
can estimate
\begin{equation} \label{recurrent}
\log_2 \mathcal N_{\beta \eps}(\B_{k-1},W_{k-1}) \leq \log_2 \kappa
+ \log_2 \mathcal N_{\eps}(\B_{k},W_{k}).
\end{equation}
We then learn from (iv) that  $\mathcal N_{Q_1} ( \B_k,W_k)=1$ for
every $k \geq 0$, and,  using \eqref{recurrent} $k$ times starting
from $\B_k$, obtain that
$$
\log_2 \mathcal N_{\beta^k Q_1}(\B_0,W_{0}) \leq k\log_2\kappa  +
\log_2 \mathcal N_{Q_1}(\B_{k},W_{k}) = k\log_2\kappa.
$$
Hence, if $\eps>0$ is arbitrary and $k$ is chosen so that
$\beta^kQ_1 \leq \eps \leq \beta^{k-1} Q_1$, we have
$$
\log_2 \mathcal N_{\eps}(\B_0,W_{0}) \leq \log_2 \mathcal
N_{\beta^kQ_1}(\B_0,W_{0}) \leq   k\log_2 \kappa \leq
\big(\log_2\textstyle \frac 1\beta\big)^{-1}\log_2 \big( \textstyle
\frac{Q_1}{\beta\eps} \big) \log_2 \kappa.
$$
Dividing by $\textstyle \log_2 \frac1\eps,$ rearranging and letting $\eps \to 0$ yields the claim.\end{proof}

We are now ready to state and prove the main result of the
section, that is, an upper bound on the fractal dimension of  the
sections $\A(t)$ of the global attractor $\A$ constructed in
Sect.\ 3.

\begin{theorem} \label{fract:thm} Assume, in addition to \emph{(H0)-(H1)}, that $\varphi \in \C^2 (\R)$ and
\begin{equation} \label{phisecond}
|\varphi''(y)| \leq c(1+|y|^\vartheta), \qquad \vartheta= \max\{q-3,0\}.
\end{equation}
Then, the global attractor $\A=\{\A(t)\}_{t \in \R}$ has finite fractal
dimension in $X_t$ for every $t \in \R$, i.e.\
$$
\dim_{X_t} \A(t)\ \leq h_1(t), \qquad  \forall  t \in \R.
$$
Moreover, the positive increasing function $h_1$ depends only on
the physical parameters of the problem  and can be explicitly
computed.
\end{theorem}

\subsection*{Proof of Theorem \ref{fract:thm}.}
 To be in  position to apply Lemma~\ref{lemma:frac}, we need to  establish a suitable smoothing property  for the process $S(\cdot,\cdot)$ restricted to the family $\A$. To do so, we need to  separate the mean value of the solution $u=u(x,t)$ on $(0,1)$.  From now on, we write $\hat f $ to indicate the mean value of $f \in L^1(0,1) $ and use the notation $\tilde f = f-\hat f$.
\vskip2mm Let  $s \in \R$, and $z=(u_0,v_0) \in \A(s)$. We
decompose $S(t,s)z $ as
$$
S(t,s) z = (u(t),\pt u (t)) = \widehat S_z(t,s) + \widetilde
S_z(t,s)  , \quad t \geq s,
$$
where $\widehat S_z(t,s) =(\widehat u(t),\widehat{ \pt u} (t))$,
and  $\widetilde S_z(t,s) = S(t,s)z -\widehat S_z(t,s).$

Denote with $\widetilde {L^2}$ [resp.\ $\widetilde{
{H}^i_{\textrm{per}}}$, $i=1,2$]  the subspace of functions of $L^2(0,1)$
[resp.\ $H^i_{\textrm{per}}(0,1)$] with zero  mean value. We will
look at the evolution of the zero mean part of the solution in the
families of Banach spaces ($t \in \R$)
\begin{align*} &\widetilde X_t = \widetilde{ {H}^1_{\textrm{per}}} \times {\widetilde {L^2}}, & & \|(u,v)\|_{\widetilde X_t}^2 =
\e^{-2Ht}|A^{1/2}u|^2+|v|^2,
\\
&\widetilde Y_t = \widetilde{{H}^2_{\textrm{per}}}  \times\widetilde{ {H}^1_{\textrm{per}}} , &
& \|(u,v)\|_{\widetilde Y_t}^2  = \e^{-2Ht}|Au|^2+|A^{1/2}v|^2,
\end{align*}

\begin{lemma} \label{lemma:mean} Let $t_0 \in \R,t_\star>0$ be fixed,  $s\leq t_0-t_\star$, $z^1,z^2 \in \A(s)$. We have the estimate
\begin{equation} \label{smooth:mean}
\|\widehat S_{z^1} (s+t_\star,s) - \widehat S_{z^2} (s+t_\star,s)
\|^2_{\R^2} \leq c\e^{C_0t_\star}\|  \hat z^1- \hat z^2
\|_{\R^2}^2,
\end{equation}
where  $C_0>0$ depends only on the physical parameters of the
problem.
\end{lemma}

\begin{proof}
Fix $s \leq t_0-t_\star$ and assume $t \in[s,s+t_\star]$. For $\imath=1,2$, let  $
z^\imath=(u_0^\imath,v_0^\imath) \in \A(s)$, and
$S(t,s)z^\imath=(u^\imath(t),\pt u^\imath(t)).$ We begin by
observing that the difference $  (m (t), m'(t))= \widehat S_{z^1}
(t,s) - \widehat S_{z^2} (t,s)$ fulfills the Cauchy problem on
$(s,s+t_\star)$
$$
m''(t) + H  m'(t)  + g_1(t)m(t)=0,\qquad m(s)= \hat u_{0}^1-\hat
u_0^2,\; m'(s)= \hat v_{0}^1-\hat v_0^2,
$$
with $$g_1(t) = \int_0^1 \varphi'(\xi(x,t)) \d x,\quad \min\{u^1(x,t),u^2(x,t)\} \leq \xi(x,t) \leq \max\{u^1(x,t),u^2(x,t)\} . $$ Now, set
$\Lambda_{m} = (m')^2 + \frac{H^2}{4}m^2 +\frac H2 mm' $. It is
immediate to determine the differential inequality
$$
\ddt \Lambda_m = -\frac 32 H (m')^2-2g_1m\big(m'+\frac H4 m\big) \leq c|g_1|(m^2 +
(m')^2) \leq c(1+R_{\AA}^{\frac{q-2}{q}}) \Lambda_m.
$$
Here, using (H1) and \eqref{boundX}, we have written
$$
|g_1(t)|  \leq \|\varphi'(\xi(t))\|_{L^{\frac{q}{q-2}}} \leq
c(1+\|u^1(t)\|^{q-2}_{L^q} +\|u^2(t)\|^{q-2}_{L^q}) \leq
c(1+R_{\AA}^{\frac{q-2}{q}}) ).
$$
An application of Gronwall's lemma on $(s,t)$ then leads to
\begin{equation} \label{smooth:meanstar}
|m'(t)|^2+ |m(t)|^2 \leq c\Lambda_m(t) \leq c\e^{C_0 (t-s)}
\Lambda_m(s) \leq  c \e^{C_0 t_\star }\|  \hat z^1- \hat z^2
\|_{\R^2}^2,
\end{equation}
for $C_0=c(1+R_{\AA}^{\frac{q-2}{q}})$. Finally, writing
\eqref{smooth:meanstar} for $t=s+t_{\star}$ gives
\eqref{smooth:mean}.
\end{proof}

\begin{lemma} \label{frac:dec}
Let $t_0 \in \R,t_\star>0$ be fixed,  $s\leq t_0-t_\star$. There
exists a decomposition
$$
\widetilde S_z(t,s) = P(t,s)[z] + N(t,s)[z], \qquad z \in \A(s),
$$
satisfying, for every $z^1,z^2 \in \A(s)$,
\begin{equation} \label{smooth:decay}
\|P(s+t_\star,s)[z^1] - P(s+t_\star,s)[z^2] \|^2_{\widetilde
X_{s+t_\star}} \leq \e^{-2Ht_\star}\|  \tilde z^1- \tilde z^2
\|_{\widetilde X_{s}}^2
\end{equation}
and
\begin{equation} \label{smoooth}
\|N(s+t_\star,s)[z^1] - N(s+t_\star,s)[z^2] \|^2_{\widetilde
Y_{s+t_\star}} \leq  C_{t_0, t_\star} \big[\|  \hat z^1- \hat z^2
\|_{\R^2}^2 +\|  \tilde z^1- \tilde z^2 \|_{\widetilde X_{s}}^2
\big]
\end{equation}
where  $C_{t_0,t_\star}>0$ depends only on $t_0,t_\star$ and on
the physical parameters of the problem.
\end{lemma}

\begin{proof}We will again use the shorthand $\varpi(t)=\e^{-2Ht}$, and repeatedly exploit the inequality
\begin{equation} \label{agmonfinal}
\|u \|^2_{L^\infty} \leq c |u||A^{1/2} u| \leq c R_\AA \varpi^{-1/2},
\end{equation}
which holds for every trajectory $u=u(\cdot)$ on the attractor $\A(\cdot)$. Here, $R_\AA$ is the radius of the absorbing set specified in \eqref{radius}.

  Throughout, assume  $t \in [s,s+t_\star]$.
We decompose $$ \widetilde S_z(t,s)= P(t,s)[z] + N(t,s)[z] =  (
p(t),  \pt p(t) ) + (n(t), \pt n(t) ) $$ where
\begin{equation}
\label{SYS-TILDE1}
\begin{cases}
\displaystyle
\ptt p(t)+H  \pt p (t) + \e^{-2Ht} A  p(t)= 0,  \\
p(s)= \tilde u_0  , \pt  p(s) =  \tilde v_0,
\end{cases}
\end{equation}
and
\begin{equation}
\label{SYS-TILDE2}
\begin{cases}
\displaystyle
\ptt  n(t)+H  \pt   n(t) + \e^{-2Ht} A    n(t)= \widehat{\varphi(u(t))} - \varphi(u(t)),   \\
 n(s)= 0  , \pt  n(s) = 0.
\end{cases}
\end{equation}
The usual multiplication of \eqref{SYS-TILDE1} by  $\pt p$ and
Gronwall's lemma on $(s,t)$ give
\begin{equation} \label{expdecayfd}
\| P(t,s)[z]\|_{\widetilde X_{t}}^2 \leq
\e^{-2H(t-s)}\|  \tilde z \|_{\widetilde X_{s}}^2.
\end{equation}
Using that $z \mapsto P(t,s)[z]$ is linear, \eqref{smooth:decay} follows from \eqref{expdecayfd} written for $t=s+t_\star,z=  z^1 -  z^2$. \vskip2mm \noindent We turn to the
difference $(\bar n(t), \pt \bar n (t)) =  N(t,s)[z^1] -
N(t,s)[z^2]$, which is a solution to
\begin{equation}
\label{SYS-TILDE}
\begin{cases}
\displaystyle
\ptt \bar n(t)+H  \pt \bar n(t) + \e^{-2Ht} A \bar  n(t)= g_2(t)+g_3(t),   \\
\bar n(s)= 0 , \pt \bar n(s) =  0.
\end{cases}
\end{equation}
with (here, $\xi$  is as above) $$
g_2= \widehat{\varphi(u^1)}-\widehat{\varphi(u^2)}= \int_0^1 \varphi'(\xi(x,\cdot))(u^1(x,\cdot)-u^2(x,\cdot)) \, \d x  , \quad g_3=
-(\varphi(u^1)-\varphi(u^2)).$$
\noindent
 We first multiply \eqref{SYS-TILDE} by $\pt \bar n$ and obtain the differential inequality
\begin{equation}
\label{SYS-TILDE0} \ddt  \| (\bar n,\pt \bar n)\|^2_{\widetilde
X_t} + 2H | \pt \bar n|^2 \leq  2 \l g_2+g_3, \pt \bar n \r\leq
c\big[|g_2|^2+ |g_3|^2\big] + 2H | \pt \bar n|^2 .
\end{equation}
The nonlinear terms are estimated as
\begin{equation} \label{rotti}
| g_2|^2    \leq  |\varphi'(\xi)|^2 |u^1-u^2|^2  \leq
  \big(1+\|u^1\|_{L^{\infty}}^{2(q-2)} + \|u^2\|_{L^{\infty}}^{2(q-2)}\big) |u^1-u^2|^2,
\end{equation}
and
\begin{align}
 | g_3|^2   & = | \varphi'(\xi)(u^1-u^2) |^2  \leq \|\varphi'(\xi)\|_{L^\infty}^2 |u^1-u^2|^2 \nonumber \\ &\leq  c \big(1+\|u^1\|_{L^{\infty}}^{2(q-2)} + \|u^2\|_{L^{\infty}}^{2(q-2)}\big)|u^1-u^2|^2.   \label{rott}
\end{align}
Therefore, writing
$
u^1-u^2= m + \bar p + \bar n,
$
where
$(\bar p(t),\pt \bar p(t))=P(t,s)[z^1] -
P(t,s)[z^2],$ and exploiting \eqref{agmonfinal} to control $\|u^i\|_{L^\infty}$, we collect the above estimates into
\begin{align*} |g_2|^2+| g_3|^2 & \leq c\big(1+R_{\AA}^{q-2}\varpi^{-q/2+1}\big) \big(m^2 + | \bar p|^2 + |  \bar n|^2\big) \\ & \leq  c\big(1+R_{\AA}^{q-2}\varpi^{-q/2+1}\big)\big( m^2 + |A^{1/2} \bar p|^2 + | A^{1/2} \bar n|^2 \big)
\\
& \leq c\big(1+R_{\AA}^{q-2}\varpi^{-q/2}\big)\big[ m^2 +\| (\bar p,\pt \bar p)\|^2_{\widetilde X_t} +\| (\bar n,\pt \bar n)\|^2_{\widetilde X_t}\big]
\end{align*}
Therefore, setting  $
C_1(t)=(1+R_{\AA}^{q-2}\varpi^{-q/2}(t)\big)$, and observing that
$C_1(t)\leq C_1({t_0})$,  \eqref{SYS-TILDE0} turns into
\begin{align}
\label{SYS-TILDE00}   \ddt  \| (\bar n,\pt
\bar n)\|^2_{\widetilde X_t}& \leq C_1(t_0) \big[m^2 +\| (\bar
p,\pt \bar p)\|^2_{\widetilde X_t}+\| (\bar n,\pt \bar
n)\|^2_{\widetilde X_t}\big]\\& \nonumber \leq cC_1(t_0)\e^{C_0
t\star} \big[\|\hat z^1 - \hat z^2 \|_{\R^2}^2 + \| \tilde
z^1-\tilde z^2\|^2_{\widetilde X_s} + \| (\bar n,\pt \bar
n)\|^2_{\widetilde X_t}\big] ,
\end{align}
where we used \eqref{smooth:meanstar} and \eqref{expdecayfd} in
the last passage. We apply Gronwall's lemma on $(s,t)$, and obtain
the intermediate estimate
\begin{equation} \label{INTER-TILDE}
\| (\bar n(t),\pt \bar n(t))\|^2_{\widetilde X_t} \leq \e^{t_\star
C_2(t_0,t_\star)} \big[\|\hat z^1-\hat z^{2} \|_{\R^2}^2 +
\|\tilde z^1-\tilde z^{2} \|_{\widetilde X_{s}}^2], \qquad t \in [
s,s+t_\star],
\end{equation}
with $C_2(t_0,t)=\exp\big(cC_1(t_0)\e^{C_0 t\star}\big)$. Now, a
further multiplication of \eqref{SYS-TILDE} by $A \pt \bar n$
yields the differential inequality
\begin{equation}
\label{SYS-TILDE12} \ddt \| (\bar n,\pt \bar n)\|^2_{\widetilde
Y_t} + 2H |A^{1/2} \pt \bar n|^2 \leq  2 \l g_2+g_3,A \pt \bar n
\r \leq  c|A^{1/2}g_3|^2 + 2H |A^{1/2} \pt \bar n|^2,
\end{equation}
due to the fact that $g_2$ is independent of $x$.
For the remaining nonlinear term, we use the mean value theorem ($\varphi \in \C^2(\R)$) and write
$$
\partial_x g_3 =  \varphi'(u^1)\partial_x u^1 - \varphi'(u^2)\partial_x  u^2 = \varphi'(u^1) \partial_x (u^1-u^2) + \varphi''(\eta) (u^1-u^2) \partial_x u^2,
$$ with $\eta=\eta(x,t)$ between $u^1(x,t)$ and $u^2(x,t)$.
Thanks to \eqref{agmonfinal}, we obtain the controls
\begin{align*}
|\varphi'(u^1) \partial_x (u^1-u^2)|^2  &
\leq c\big(1+\|u^1\|_{L^\infty}^{2(q-2)}\big) \big[|A^{1/2} \bar
p|^2+|A^{1/2} \bar n|^2\big]  \\ & \leq c\big(1+R_{\AA}^{q-2}\varpi^{-q/2+1}\big) \big[|A^{1/2} \bar
p|^2+|A^{1/2} \bar n|^2\big]
\end{align*}
and, recalling \eqref{phisecond},
\begin{align*}
|\varphi''(\eta) (u^1-u^2) \partial_x   u^2|^2 & \leq c(1+ \|u_1\|^{2\vartheta}_{L^\infty}+ \|u_2\|^{2\vartheta}_{L^\infty})|u^1-u^2|^2 |A^{1/2} u^2 |^2 \\&  \leq  c\big(1+R_{\AA}^{1+\vartheta}\varpi^{-1-\frac{\vartheta}{2}}\big) \big(m^2 + |A^{1/2} \bar p|^2 + |A^{1/2}  \bar n|^2\big).
\end{align*}
Being $1+\vartheta \leq q-2,$ and $\tilde\vartheta:=1+\frac{\vartheta}{2} > \frac q2-1$, we summarize the above bounds into
\begin{align*}
  | A^{1/2}g_3|^2   & \leq c |\partial_x g_3|^2  \leq c\big(1+R_{\AA}^{q-2}\varpi^{-1-\tilde \vartheta} \big)\big[\| (\bar p,\pt \bar p)\|^2_{\widetilde X_t}+ \| (\bar n,\pt \bar n)\|^2_{\widetilde X_t} + m^2\big].
\end{align*}

 Setting  $
C_{3}(t)=c\big(1+R_{\AA}^{q-2}\varpi^{-1-\tilde \vartheta}\big) $,
and again observing that $C_3(t)\leq C_3({t_0})$,
\eqref{SYS-TILDE12}  turns into
\begin{align}
\label{SYS-TILDE3}   \ddt \| (\bar n,\pt
\bar n)\|^2_{\widetilde Y_t}& \leq C_3({t_0}) \big[\| (\bar p,\pt
\bar p)\|^2_{\widetilde X_t}+ \| (\bar n,\pt \bar
n)\|^2_{\widetilde X_t} + m^2\big] \\ \nonumber & \leq
C_3({t_0}) \e^{C_2(t_0,t_\star) t_\star} \big[ \|\tilde z^1-\tilde
z^{2} \|_{\widetilde X_s}^2 +\|\hat z^1-\hat z^{2}
\|_{\R^2}^2\big].
\end{align}
Here we used again \eqref{smooth:mean}, \eqref{smooth:decay} and
\eqref{INTER-TILDE} in the last passage. Finally, we integrate on
$(s,s+t_\star)$. Being $\bar n(s)=0, \pt \bar n(s)=0$, we end up
with
$$
\| (\bar n(s+t_\star),\pt \bar n(s+t_\star))\|^2_{\widetilde
Y_{s+t_\star}} \leq  C_{t_0,t_\star}\big[ \|\tilde z^1-\tilde
z^{2} \|_{\widetilde X_s}^2 + \|\hat z^1-\hat z^{2}
\|_{\R^2}^2\big],
$$
which is \eqref{smoooth}, with $C_{t_0,t_\star}=t_\star C_3({t_0})
\e^{C_2(t_0,t_\star)}$.
\end{proof}

We are ready to complete the proof of Theorem \ref{fract:thm}. To
begin with, we identify each $z \in \A(t)$ with the pair $ (\hat
z, \tilde z) \in \R^2 \times \widetilde X_t$. We then define the
family
$$
\A_\star=\Big\{\A_\star(t) =\{(\hat z,\tilde z): z \in \A(t)\}, {t
\in \R}\Big\}, \qquad \A_\star(t)  \subset \R^2 \times \widetilde
X_t.
$$
From the properties of $\A$ (see Theorem~\ref{osc:attractor}),
each $\A_\star(t)$ is bounded in $\R^2 \times \widetilde Y_t$ and
hence compact in $\R^2 \times \widetilde X_t$. The bound
\eqref{boundX}  also guarantees that $\|\A_\star(t)\|_{\R^2 \times
\widetilde X_t} \leq Q_1$, for some positive $Q_1$ depending only
on $R_{\AA}$.

Now, fix $t_0 \in \R$ and set  $t_\star= 3H^{-1}\log 2$. Referring
to Lemma~\ref{lemma:frac}, for $k \in \N \cup\{0\}$, set
\begin{align*} &
W_k=\R^2 \times \widetilde X_{t_0-kt_\star}, && \|(\hat z,\tilde
z) \|_{W_k}^2= \|\hat z\|^2_{\R^2}+ \|\tilde z\|_{\widetilde
X_{t_0-kt_\star}}^2,
\\ & Z_k=\R^2 \times \widetilde Y_{t_0-kt_\star}, && \|(\hat z,\tilde z) \|_{Z_k}^2= \|\hat z\|^2_{\R^2}+ \|\tilde z\|_{\widetilde Y_{t_0-kt_\star}}^2.
\end{align*}
We easily have $Z_k \Subset W_k$ for each $k$, so that assumption
(i) is verified. Then, we point out that the linear isomorphism
$$
W_k  \ni z= \big(\hat z, (\tilde u, \tilde v)\big) \mapsto
J_k(z)=\big(\hat z, (\e^{-kt_\star H}\tilde u, \tilde v)\big) \in
W_0
$$
satisfies $J_k(B_{W_k}) = B_{W_0}, J_k(B_{Z_k})=B_{Z_0}$. Hence,
for $\eps >0$,
$$
\mathcal N_{\eps}\big(B_{Z_{k}},W_{k}\big) = N_{\eps} \big
(B_{Z_0},W_0) = \kappa_\eps < \infty,
$$
so that (ii) is satisfied as well. Finally, we consider the sets
$\B_k=\A_\star(t_0-kt_\star)$, for which Properties (iii)-(iv)
have been already verified above. For $k \geq 1$, define the maps
$U^k: \B_{k} \to \B_{k-1}$,
$$
z=(\hat z,\tilde z)  \mapsto  U^k\big((\hat z, \tilde z)\big) =
\big(\widehat S_{z}(t_0-(k-1)t_\star,t_0-kt_\star), \widetilde
S_{z}(t_0-(k-1)t_\star,t_0-kt_\star)\big) ,
$$
Due to the invariance of $\A$, the $U^k$ are well defined and
$U^k(\B_k)=\B_{k-1},$ which is (v). Regarding (vi),
referring to Lemma \ref{frac:dec}  we write  $U^k(z)=P^k(z) + N^k(z),$ with
\begin{align*} &
P^k(z) = \big(0,P(t_0-(k-1)t_\star,t_0-kt_\star) [z] \big),
\\ &  N^k(z)=\big(\widehat S_{z}(t_0-(k-1)t_\star,t_0-kt_\star),N(t_0-(k-1)t_\star,t_0-kt_\star) [z] \big).
\end{align*}
Writing
\eqref{smooth:decay} for $s=t_0-kt_\star$, we learn that
$$
\| P^k  (z^1) - P^k (z^2)\|_{W^{k-1}} \leq \e^{-Ht_\star}
\|z^1-z^2\|_{W_k} =\textstyle \frac18 \|z^1-z^2\|_{W_k} , \qquad
\forall z^1,z^2 \in \B_{k},
$$
so that \eqref{decay} is satisfied with $\varrho=\frac18$.
Moreover, collecting \eqref{smooth:mean} and \eqref{smoooth}
written for $s=t_0-kt_\star$, we have
$$
\| N^k  (z^1) - N^k (z^2)\|_{Z^{k-1}} \leq Q_{t_0}
\|z^1-z^2\|_{W_k}  \qquad \forall z^1,z^2 \in \B_{k},
$$
with $Q_{t_0}=(c\e^{C_0t_\star}+ C_{t_0, t_\star})^{1/2},$ which
is \eqref{smoothing}. Therefore, Lemma~\ref{lemma:frac} applies, yielding
\begin{equation} \label{boundfd}
\dim_{\R^2 \times \widetilde X_{t_0}} \A_\star(t_0)= \dim_{W_0}
\B_0 \leq    \log_2 \kappa_{(4Q_{t_0})^{-1}}.
\end{equation}
To complete the proof, it is enough to recall that, for fixed $t$,
$\R^2 \times \widetilde X_{t}$ and $X_{t}$ are isomorphic as
Banach spaces through the map
$$
\R^2 \times \widetilde X_{t} \ni (\hat z,\tilde z) \mapsto J_{t}
(\hat z, \tilde z) = \hat z + \tilde z \in X_t.
$$  Indeed, $J_t$ is clearly bijective and, if $\big(\hat z,\tilde z\big)=(\big(\hat u,\hat v),(\tilde u, \tilde v)\big), $
\begin{align*}
\|J_{t}(\hat z, \tilde z) \|_{X_{t}}^2& \leq 3 \Big[ \|\hat u+\tilde
u\|_{L^q}^2+e^{-2Ht}|A^{\frac12} \tilde u|^2 + |\hat v+\tilde
v|^2\Big] \leq 3\Big[ \|\hat z \|_{\R^2}^2 + \|\tilde u\|_{L^\infty}^2
  + \|\tilde z\|^2_{\widetilde X_{t}} \Big] \\ &
 \leq c \Big[\|\hat z \|_{\R^2}^2 +\lambda_1^{-1}|A^{1/2} \tilde u|^2+\|\tilde z\|^2_{\widetilde X_{t}} \Big]  \leq c(1+\lambda_1^{-1}\e^{2Ht})\|(\hat z,\tilde z) \|_{\R^2 \times \widetilde X_t}^2.
\end{align*}
Since $J_{t_0}(\A_\star(t_0))= \A(t_0)$ and Banach space
isomorphisms preserve  fractal dimension, \eqref{boundfd} implies
$$
\dim_{X_{t_0}}\A(t_0) =\dim_{\R^2 \times \widetilde X_{t_0}}
\A_\star(t_0)
 \leq   \log_2 \kappa_{(4Q_{t_0})^{-1}},
$$
as well, which is the statement of Theorem \ref{fract:thm}, with
$h_1(t_0)= \log_2 \kappa_{(4Q_{t_0})^{-1}}.$

\subsection*{Acknowledgments.}  This work was partially
supported by the National Science Foundation under the grants
NSF-DMS-0604235, NSF-DMS-0906440, and by the Research Fund of Indiana University.

\end{document}